\documentclass[12pt]{amsart}

\usepackage{amssymb}
\usepackage{amsmath}

\usepackage[dvips]{graphicx}
\newtheorem{theorem}{Theorem}
\newtheorem{lemma}{Lemma}

\date{April 10, 2018}
\usepackage{subfig}
\usepackage{color}
\usepackage[margin=1in]{geometry}

\newtheorem*{theorem*}{Theorem}{\bf}{\it}
\newtheorem*{proposition*}{Proposition}{\bf}{\it}
\newcommand{\tP}{{\tilde P}}
\newcommand{\tQ}{{\tilde Q}}

\begin{document}
\title[Search for cycles]{Search for cycles in non-linear autonomous discrete dynamical systems}


\author{D. Dmitrishin,  A. Stokolos and M. Tohaneanu}
\maketitle

\begin{abstract} We construct a family of polynomials with real coefficients that contains as a particular case the Fej\'er and Suffridge polynomials. These polynomials allow us to suggest a robust algorithm to search for cycles of arbitrary length in non-linear autonomous discrete dynamical systems. Numeric examples are included.
\end{abstract}

\section{Introduction}

\subsection{Settings}
Consider the discrete dynamical system 
\begin{equation}\label{system} 
x_{n+1}=f(x_n),\qquad  f: A\to A, \; A\subset \mathbb R^m
\end{equation} 
where $A$ is a convex set that is invariant under $f$. Let us assume that the system has an unstable  T-cycle  $(x_1^*,...,x_T^*)$. We define the cycle multipliers $\mu_1,...,\mu_m$ as the zeros of the characteristic polynomial
\begin{equation} \label{cheq}
{\rm det}\left( \mu I-\prod_{j=1}^T {  D}f (x_j^*) \right)=0.
\end{equation} 
We will assume that the multipliers are located in a region $M\subset \mathbb C.$ If all the multipliers are inside the unit disc $\mathbb D = \{|z|<1\}$ in the complex plane, then the cycle is locally asymptotically stable. If not, an infinitesimal perturbation of the cycle values can lead to  behavior called ``deterministic chaos" or just ``chaos."  The term was coined by J.A.~Yorke and T.Y.~Li  in the paper ``Period Three Implies Chaos" (1975) \cite{LY} {in which it was proved}  that any one-dimensional system which exhibits a regular cycle of period three will also display regular cycles of every other length, as well as completely chaotic cycles. The famous Sharkovsky's theorem \cite{Sha}  includes this result as a special case (cf.\cite{BB}, p. 79).

Various methods of chaos control have been developed since, starting with the groundbreaking work of Edward Ott, Celso Grebogi and James A. Yorke \cite{OGY}  (1990). In \cite{OGY} the authors suggest a method to stabilize chaos by making only small time-dependent perturbations of an available system parameter. They also show that small time dependent changes in the control parameters of a nonlinear system can turn a previously chaotic trajectory into a stable, periodic motion.

The next step forward was done by Kestutis Pyragas \cite{P} in 1992. He suggested a very simple linear scheme 
$f(x_{n+1})=f(x_n)+K(x_n-x_{n-1}).$ The Pyragas method has turned out to be very popular because it is easy to implement experimentally. It has been used in a large variety of systems in physics, chemistry, biology, medicine, and engineering \cite{SS, P1,G} .

In 1996 Toshimitsu Ushio \cite{U} showed that the Pyragas method has several serious limitations. In particular it was shown  that the admissible region of the multiplier  is (-3,1), so it does not work for the whole range of negative values as one would like. Moreover, it was shown in \cite{DKS} that going deeper in the prehistory by adding more delays does not improve the situation.  If two multipliers of the system belong to the same connected component of the region of stability, then the distance between them will be at most 4, regardless of the number of delays.

In 1996 M. de Sousa Vieira and A.J. Lichtenberg \cite{VL} suggested a non-linear counterpart of the Pyragas method.  In their ``Conclusion and discussion" section they wrote ``The generalization consists of feeding back the nonlinear mapping signal rather than a signal linearized around the fixed point. This increases the basin of attraction of the controlled signal and thus decreases the sensitivity to noise. {\it However, the range of parameters for which control can be achieved is limited}."

Further work on the stabilization of cycles was done by J. E. S. Socolar, D. W. Sukow, and D. J. Gauthier \cite{SSG} and \"O.Morgul \cite{M,M1}.  In particular, the latter papers considered the problem of finding 2,4,5 and 6-cycles of the logistic map $f(x) =\mu x(1 - x)$ for various values of $\mu.$ 

The goal of this paper is to provide a robust method to stabilize cycles of any length whose multipliers lie in the region $(-\infty, 1)$. We will be looking at the non-linear control \eqref{ave} and find coefficients (that we conjecture to be optimal) that allow us to stabilize the cycle in polynomial time. Sections 2-4 provide a preliminary discussion and setup of the problem. In Section 5 we discuss the cases $T=1$ and $T=2$, which have been rigorously studied by the authors and collaborators in \cite{DH}, \cite{DKKS}. Sections 6-8 are devoted to discussing our choice of coefficients, and providing conjectures and experimental evidence of why we believe these conjectures to hold. In Section 9 we rigorously prove asymptotic  bounds for the size of the multipliers that can be stabilized for a cycle of length $T$ if we allow the length of prehistory $N$ to go to infinity. More precisely, we obtain a polynomial-type bound of approximately $N^{2}$ for the size of the multipliers, which is important in applications (as opposed to, say, an exponential bound). Finally, Section 10 is dedicated to numerical examples of how our method can be applied to well-known dynamical systems.

\section{Average system}

The standard approach developed in Analysis to suppress oscillations is averaging. Let us apply this idea to stabilize unstable $T$-cycles, i.e. given a range for the multipliers $\mu$ we want to stabilize the cycle by the following averaging procedure
\begin{equation}\label{ave}
x_{n+1}=\sum_{k=1}^Na_k f (x_{n-kT+T}),\qquad \sum_{k=1}^Na_k=1. 
\end{equation}
Note that the system \eqref{ave} preserves  the convex invariant set and the T-cycles of the system \eqref{system}, while offering great flexibility since we can choose the coefficients $a_j$.\\ 

Two natural questions arise:  {\it Can stabilization be obtained with a bounded $N$ - the depth of prehistory? If so, what is the minimal depth necessary ?} \\

\section{Stability analysis} 

The characteristic equation for the system \eqref{ave}  is
$$
\prod_{j=1}^m\left[\lambda^{T(N-1)+1}-\mu_j\left(\sum_{k=1}^N a_k\lambda^{N-k}\right)^T\right]=0,\quad \mu_j\in M,\; j=1,...,m.
$$
The proof for the scalar case $m=1$ is in \cite{DHKS}, and for the vector case is in \cite{Kh}. The form of the polynomial allows one to establish a nice geometric criterion (c.f. \cite[5.1]{DKST}) that was suggested  by Alexei Solyanik  \cite{Sol}.

\begin{lemma}\label{l1} 
The characteristic polynomial of the system \eqref{ave} has all the roots inside the unit disc as long as the reciprocal values of the multipliers are outside the image of the unit disc under the polynomial map 
$F_T(z)=z(a_1+...+a_Nz^{N-1})^T,$ i.e.
$$
\frac1{\mu_j}\not\in F_T(\mathbb D),\quad j=1,...,N,
$$
or
$$
\mu_j\in( \bar{\mathbb C}\backslash F_T(\mathbb D))^*,
$$
where $z^*=1/{\bar z}.$
\end{lemma}

Note that in the case $N=1$ one gets $F_T(z)=z$ and $F_T(\mathbb D)=\mathbb D$ and thus $( \bar{\mathbb C}\backslash F_T(\mathbb D))^*=\mathbb D$ which is the standard criterion for stability in the open-loop system \eqref{system}

\section{Optimization problem}

In this paper we will assume that the multipliers lie on the half-axis $(-\infty, 1)$. In this case the problem of stabilization can be reduced to the following optimization problem: find
$$
I^{(T)}_N=\sup_{\sum_{j=1}^{N}a_j=1}\min_{t\in[0,\pi]}\left\{\Re\left(F_T(e^{it}) \right): \Im\left(F_T(e^{it}) \right)=0 \right\}.
$$

Lemma~\ref{l1} now implies that  for the system \eqref{ave} a robust stabilization (i.e. by the same control for all $\mu_j\in(-\mu^*,1)$) of any $T$-cycle is  possible if 
\begin{equation}\label{mu}
(\mu^*)\cdot|I_N^{(T)}|\le1.
\end{equation}

We are left with the task of finding the polynomials that solve the optimization problem, and estimate $|I^{(T)}_N|.$ 

\section{Case $T=1,2$ - Suffridge polynomials}
For $T=1,2$ the optimization problems were solved in \cite{DH} and \cite{DKKS} by means of Harmonic Analysis. Namely, for $T=1$
\begin{equation}\label{I1N}
|I_N^{(1)}|=\inf_{\sum_{j=1}^{N}a_j=1}\max_{t\in[0,\pi]}\left\{-\sum_{j=1}^{N}a_j\cos jt:
\sum_{j=1}^{N}a_j\sin jt=0 \right\}=\tan^2\frac\pi{2(N+1)} \sim \frac{\pi^2}{4N^2}.
\end{equation}
The optimal coefficients are the coefficients of Suffridge polynomials (see \cite{S})
\begin{equation}\label{a1}
a_j^{(1)}=A_N\left(1-\frac j{N+1}\right)\sin\frac{\pi j}{N+1},\quad
A_N=2\tan\frac\pi{2(N+1)},\; j=1,...,N.
\end{equation}
Similarly for $T=2$ one has
\begin{equation}\label{I2N}
|I^{(2)}_N|=\inf_{\sum_{j=1}^{N}a_j=1}\max_{t\in[0,\pi]}
\left(
-\sum_{j=1}^{N}a_j\sin(2j-1)t: \sum_{j=1}^{N}a_j\cos(2j-1)t=0 \right)^2=\frac1{N^2}.
\end{equation}
The optimal coefficients are odd coefficients of Fej\'er polynomials
\begin{equation}\label{a2}
a_j^{(2)}=A_N\left(1-\frac {2j-1}{2N}\right),\qquad A_N=\frac2N,\; j=1,...,N.
\end{equation}
The images $F_T(e^{it})$ of the unit circle under the optimal polynomial maps $F_T(z)$ look very similar for $T=1$ and for $T=2$ (see Fig.\ref{T1N5}  and Fig.\ref{T2N6}).  The only major difference between the two pictures is the behavior when $t=\pi$: for $N$-even the graph is tangent to the x-axis there, which does not happen for $N$-odd.

\begin{figure}[!htbp]
\centering
\begin{minipage}[b]{0.55\linewidth}
\includegraphics[scale=0.25]{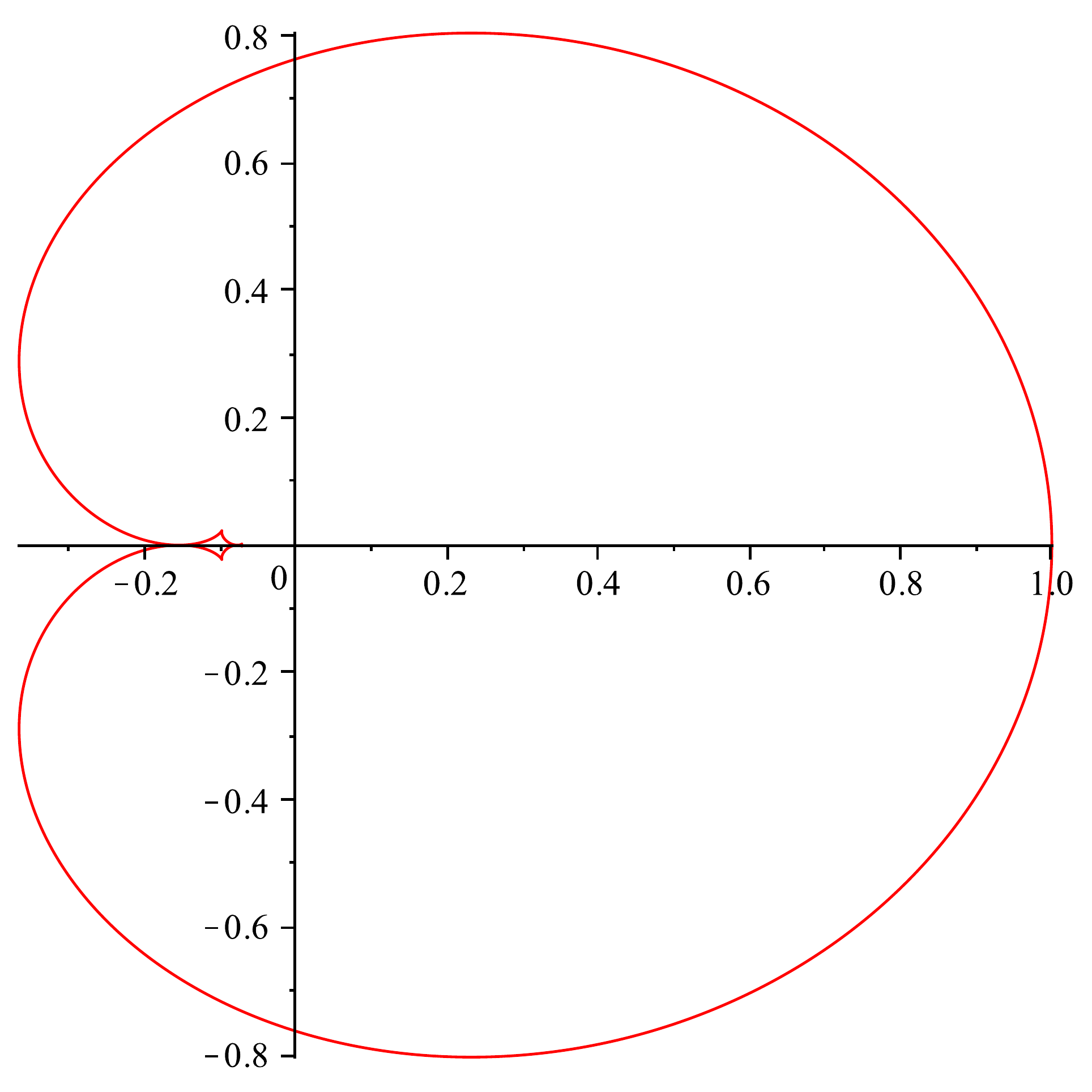}
\caption{$F_T(e^{it})$ for T=1 and N=5}
\label{T1N5}
\end{minipage}
\begin{minipage}[b]{0.35\linewidth}
\includegraphics[scale=0.25]{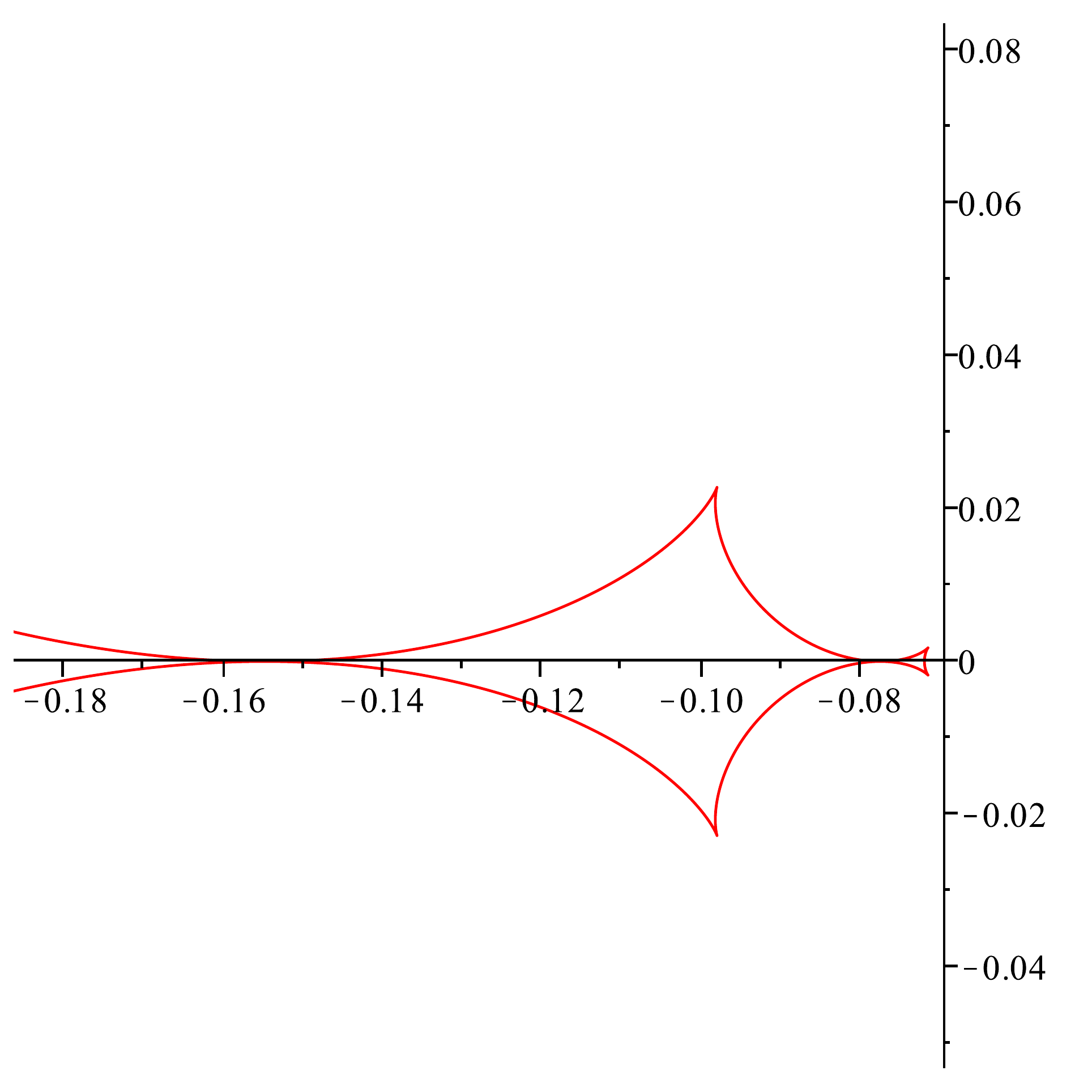}
\caption{Fragment}
\label{T1N5F}
\end{minipage}
\end{figure}

\begin{figure}[!htbp]
\centering
\begin{minipage}[b]{0.55\linewidth}
\includegraphics[scale=0.25]{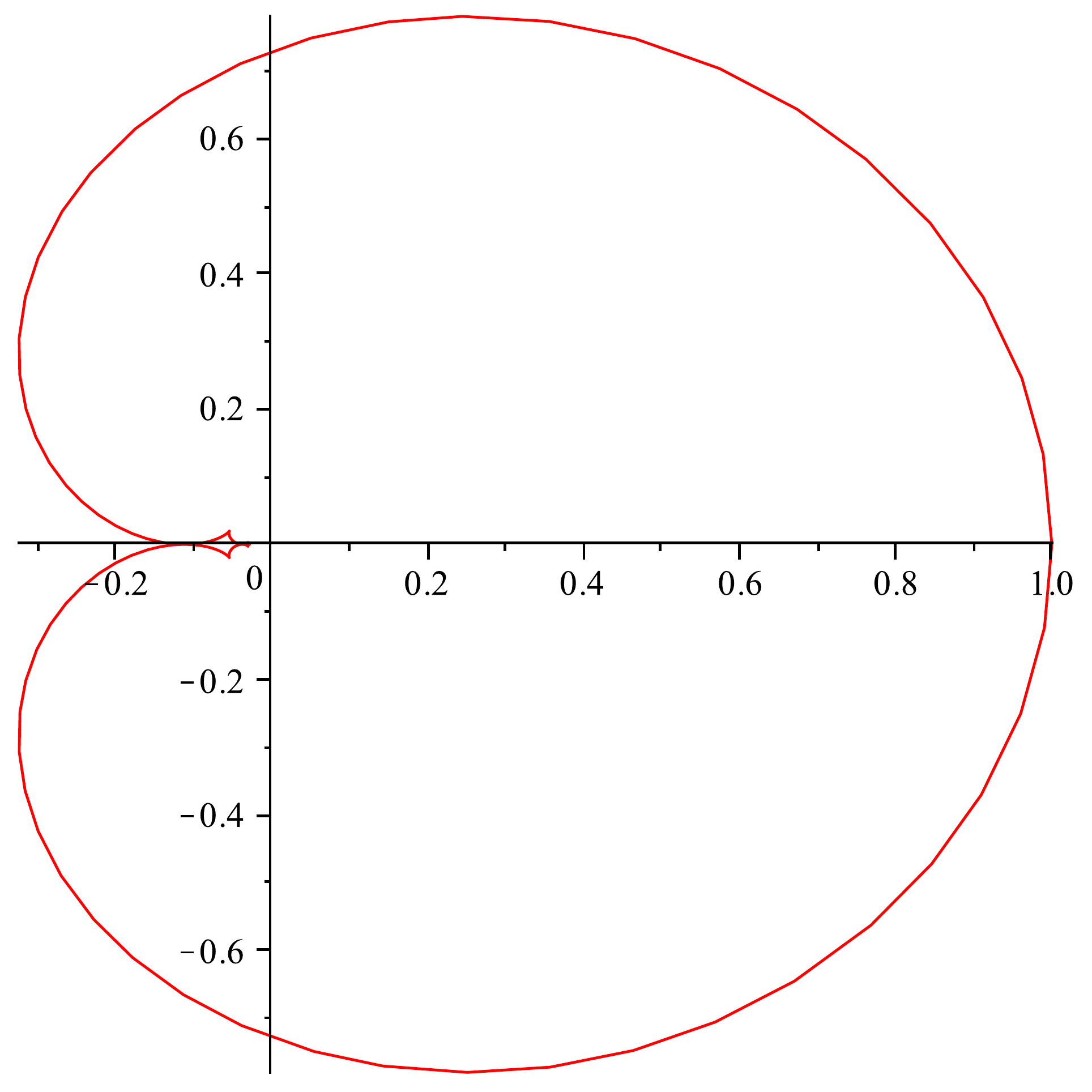}
\caption{$F_T(e^{it})$ for T=2 and N=6}
\label{T2N6}
\end{minipage}
\begin{minipage}[b]{0.35\linewidth}
\includegraphics[scale=0.25]{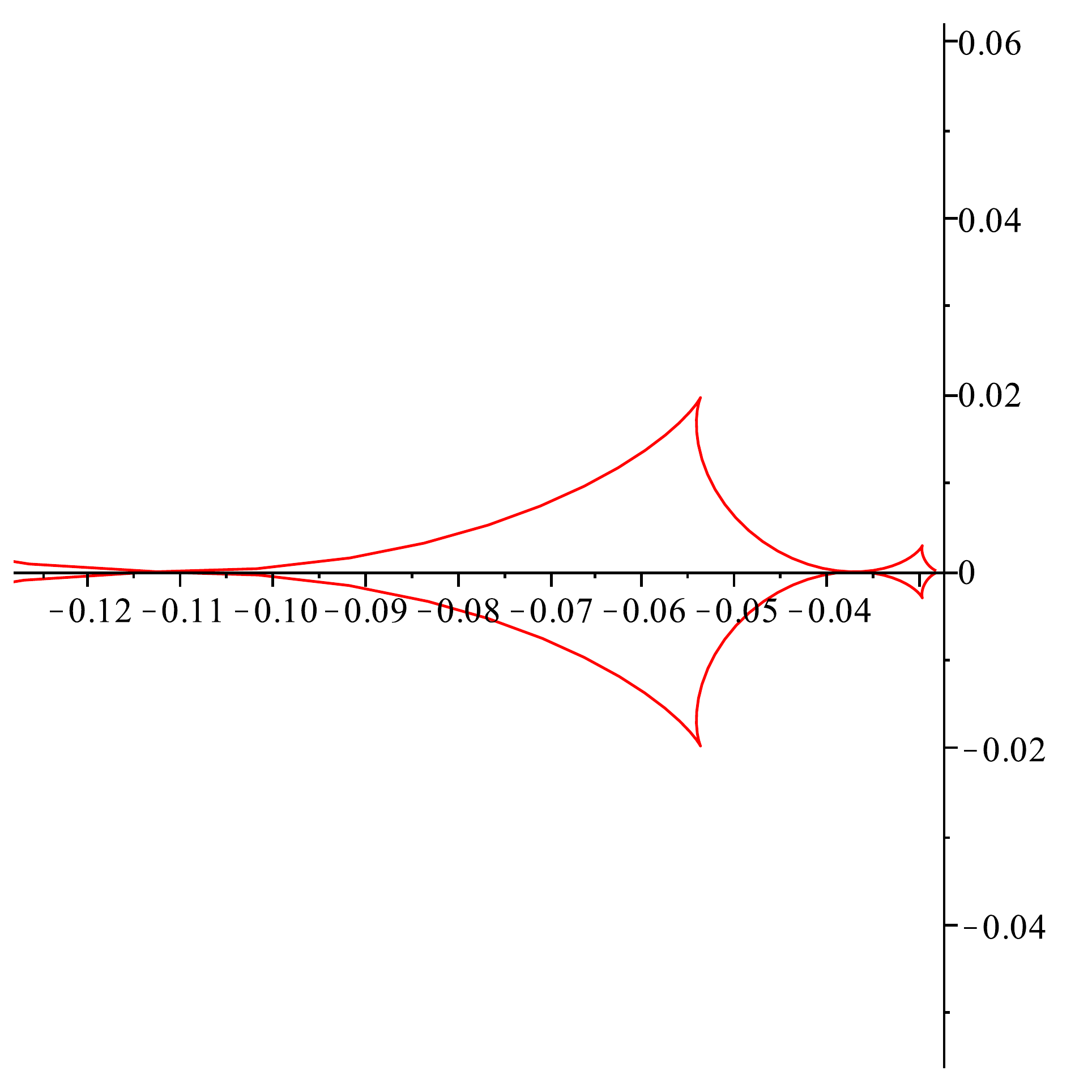}
\caption{Fragment}
\label{T2N6F}
\end{minipage}
\end{figure}

\section{Case $T\ge3$ - generalized Suffridge polynomials}
In \cite{DKST} more general family of polynomials were introduced. Those polynomials generalize Suffridge polynomials for the case $T\ge3$ and turn into the polynomials from the above section if $T=1,2.$ In this section we provide the explicit formulas for these polynomials.\\

Let $q(z):= a_1+...+a_Nz^{N-1}$. Define the set of points 
$$
\psi_j=\frac{\pi(2+T(2j-1))}{2+(N-1)T},
\quad j=1,.., \frac N2 \;\mbox{(N-even)}, \;
\left(\frac{N-1}2\;\mbox{(N-odd)}\right)
$$
and the generating polynomials
$$
\eta(z)=z(z+1)\prod_{j=1}^{\frac{N-2}2}(z-e^{i\psi_j})(z-e^{-i\psi_j}),\quad\mbox{N-even}
$$
$$
 \eta(z)=z\prod_{j=1}^{\frac{N-1}2}(z-e^{i\psi_j})(z-e^{-i\psi_j}),\quad \mbox{N-odd}.
$$
We now let
\begin{equation}\label{qz}
q(z)=\frac{KT}{2+(N-1)T}\left( \left( \frac1T+N \right)\frac{\eta(z)}z-\eta^\prime(z) \right),
\end{equation}
where $K$ is a normalizing factor so that $q(1)=1.$ A direct computation shows that
$$
\frac1K=2^{\frac{N-2}{2}}\prod_{j=1}^{\frac{N-2}{2}}(1-\cos\psi_j),\qquad\mbox{$N$ is even,}
$$
$$
\frac1K=2^{\frac{N-3}{2}}\prod_{j=1}^{\frac{N-1}{2}}(1-\cos\psi_j),\qquad\mbox{$N$ is odd,}
$$
as well as 
\begin{equation}\label{PN}
q(-1)=P_N
\end{equation}
 where
$$
P_N=  \frac T{2+(N-1)T}\prod_{j=1}^{\frac{N-2}2}\cot^2\frac{\psi_j}2,\; N \; \mbox{is even};\quad
P_N=\prod_{j=1}^{\frac{N-1}{2}}\cot^2\frac{\psi_j}2, \; N \; \mbox{is odd}.
$$
In order to compute the coefficients $a_j$ we write $\eta(z)$ in the standard form 
$$\eta(z)=z\sum_{j=1}^Nc_jz^{j-1}$$
We then have
\begin{equation}\label{ac}
a_j=\left(1-\frac{1+(j-1)T}{2+(N-1)T}\right)c_j.
\end{equation}

\section{Experimental evidence}

Using MAPLE  we plotted  the image of the unit disc $\mathbb D$ under the polynomial map $F_T(z)=z(q(z))^T$ for various combinations of $T$ and $N,$ and all the plots look remarkably similar.  Below we provide the images for a small value $N=5$ and two different values of $T$, a small one ($T=5$) and a large one ($T=1005$), see Fig~\ref{T5N5} and Fig~\ref{T1005N5}. We chose $N=5$ to be able to better observe the behavior of the curve $F_T(e^{it})$  near the point $t=\pi,$ as for large $N$ this is harder. 

A typical case of the plot is in Fig.~\ref{T5N5} below. One can observe several cusps, which suggests that the roots of the derivatives of the polynomials $F_T(z)$ are on the boundary of the unit disc $\mathbb D$, which was a crucial step in proving univalency of Suffridge polynomials. For completeness we also provide the inverse image $F_T(\mathbb D))^*$, which is the unbounded region on the right. We note in particular that the interval of multipliers $(\mu^*,1)$ does not intersect the interior of $F_T(\mathbb D))^*$, although it does intersect the boundary. This can be easily remedied by considering $F_T^{\epsilon} = \frac{F_T(z)+\epsilon z}{1+\epsilon}$ and sending $\epsilon$ to $0$, so considering the original $F_T$ is justified for computer simulations.

\begin{figure}[!htbp]
\centering
\begin{minipage}[b]{0.55\linewidth}
\includegraphics[scale=0.35]{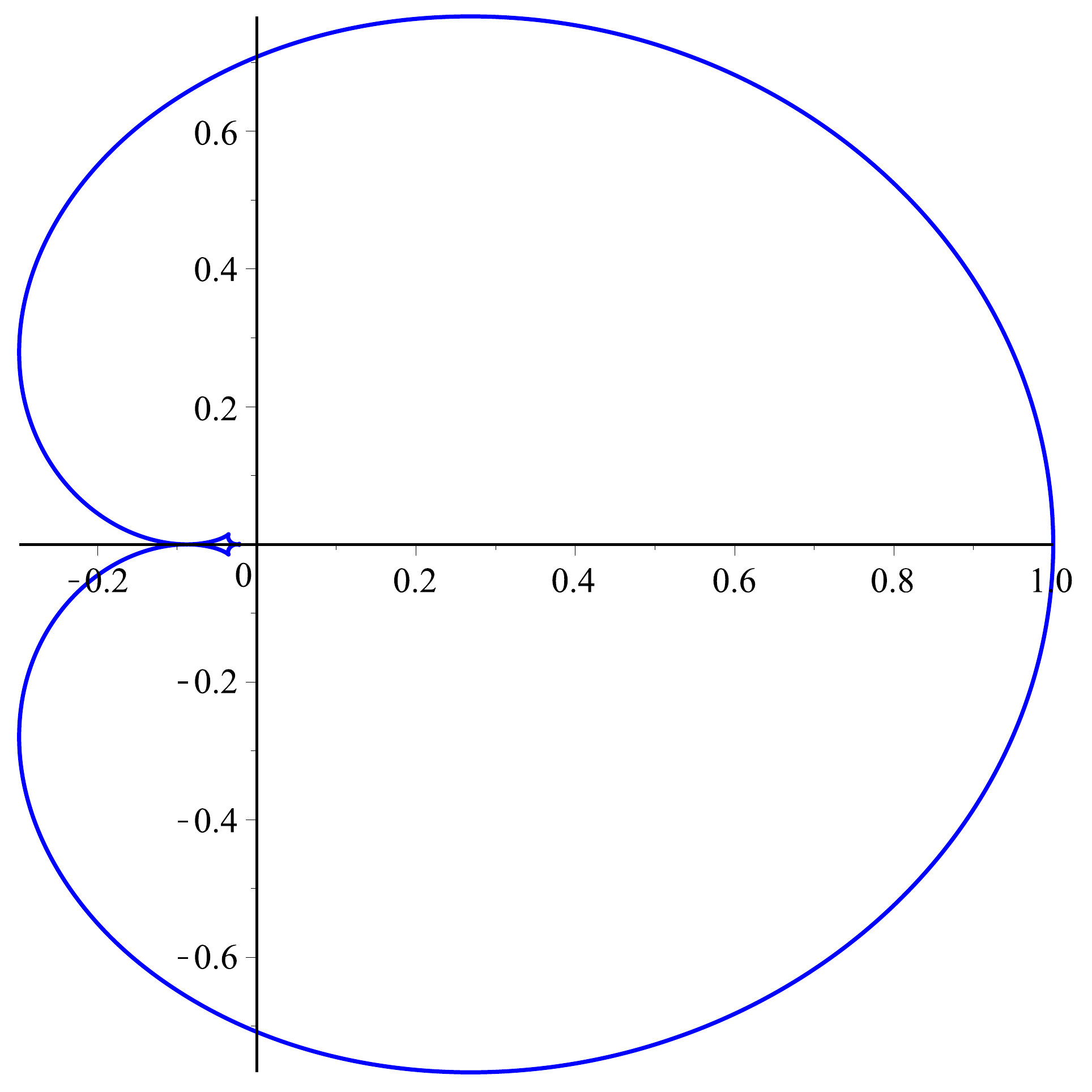}
\caption{$F_5(\mathbb D)$, N=5}
\label{T5N5}
\end{minipage}
\begin{minipage}[b]{0.55\linewidth}
\includegraphics[scale=0.35]{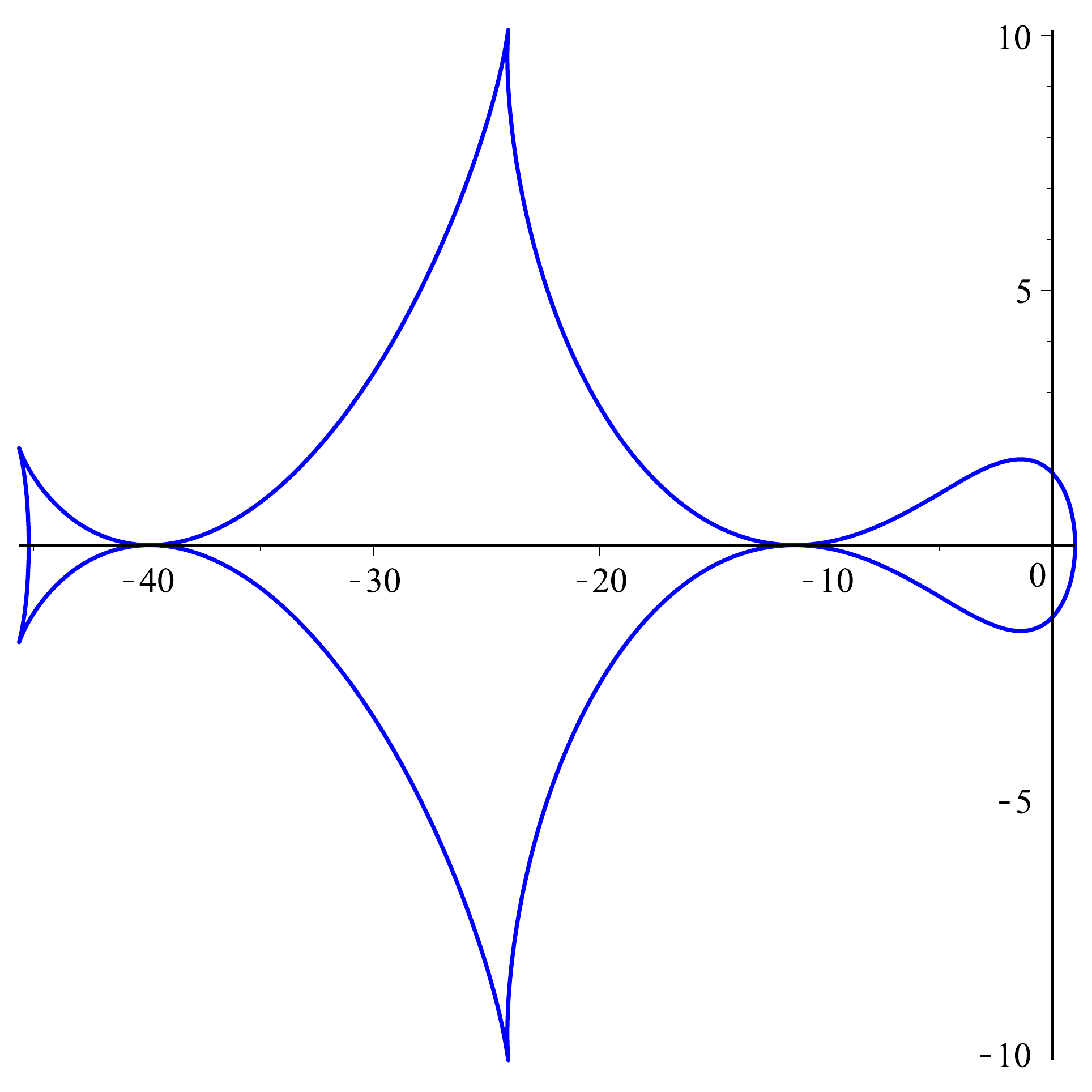}
\caption{$(\mathbb C \backslash F_5(\mathbb D))^*$,  N=5}
\label{T5N5i}
\end{minipage}
\centering
\begin{minipage}[b]{0.55\linewidth}
\includegraphics[scale=0.35]{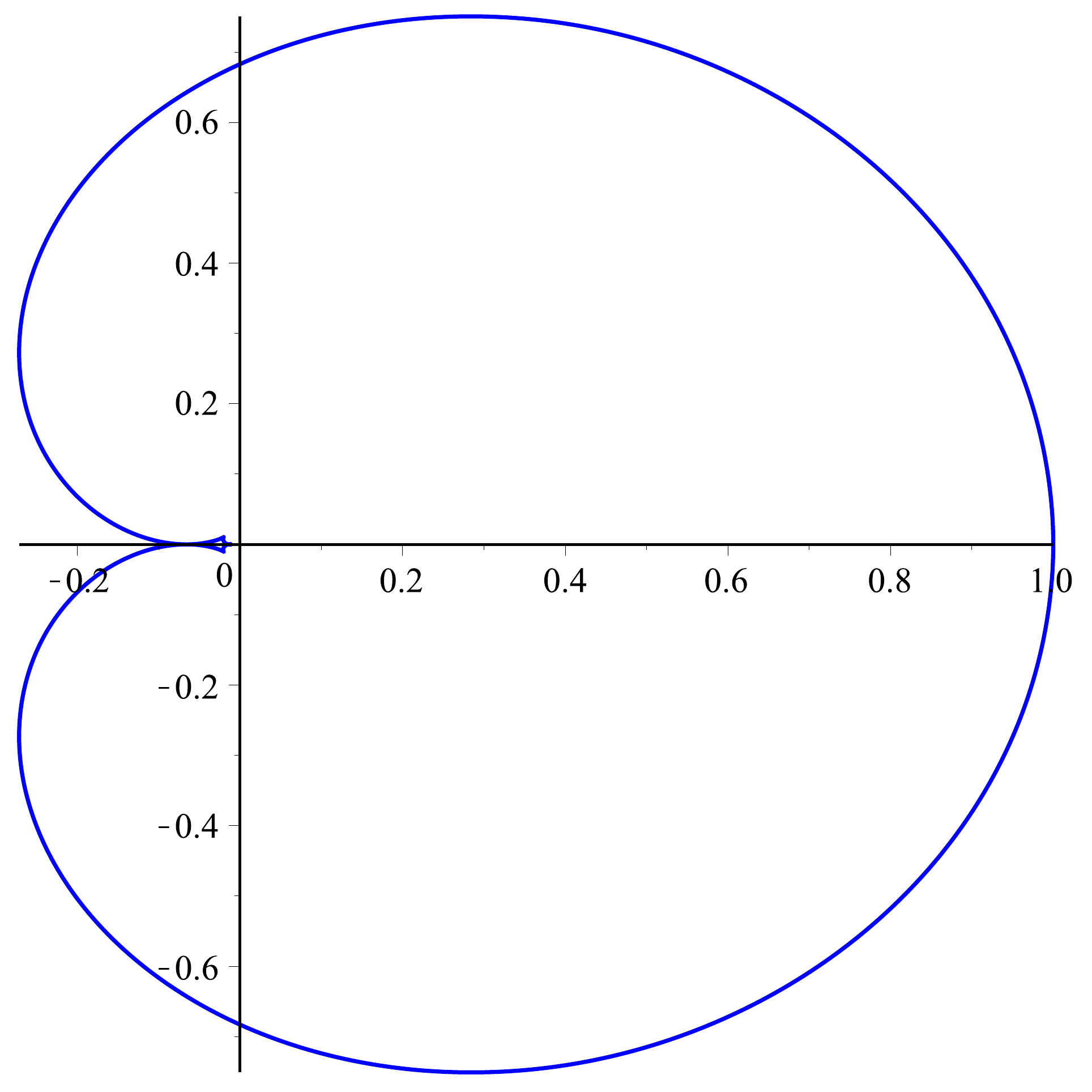}
\caption{$F_{1005}(\mathbb D)$, N=5}
\label{T1005N5}
\end{minipage}
\begin{minipage}[b]{0.55\linewidth}
\includegraphics[scale=0.35]{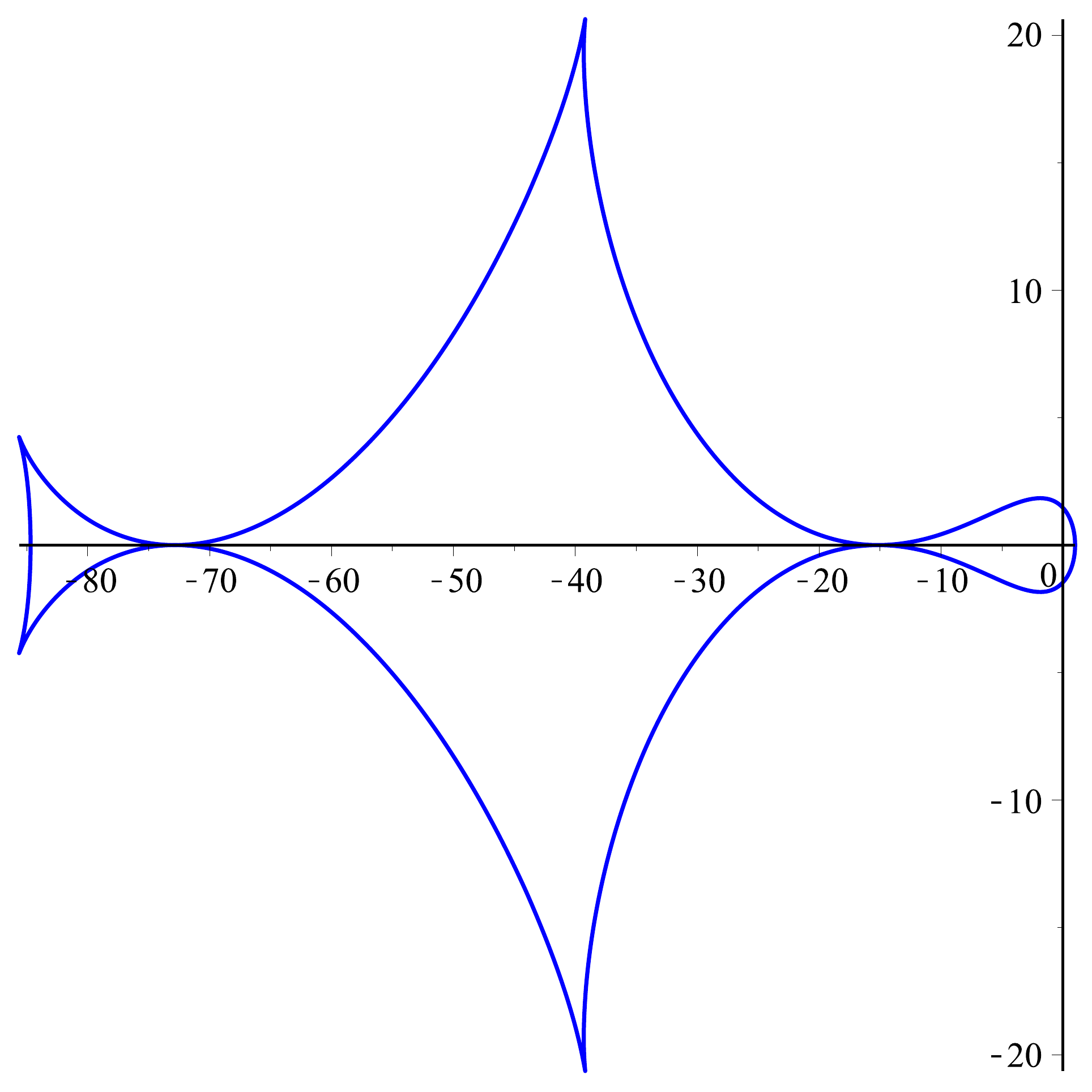}
\caption{$(\mathbb C \backslash F_{1005}(\mathbb D))^*$,  N=5}
\label{T1005N5i}
\end{minipage}
\end{figure}

Based on the observations above we can conjecture that the polynomials $F_T(z)=z(q(z))^T$ map 
the curve $z=e^{it}, 0<t<\pi$ into a curve in the closed upper half plane $\{\Im(z)\ge 0\}.$  Moreover, the image of the curve should not have any self-intersections. Thus, we conjecture\\

{\bf Conjecture A}: {\it 
The polynomials $F_T(z)=z(q(z))^T$ are univalent.\\
}

With\eqref{PN} in mind we also conjecture\\

{\bf Conjecture B}: {\it 
Let $T\geq 3$ and $N\ge 1$ be integers. Then $|I_N^{(T)}|\le (P_N)^T.$\\
}

The presence of cusps in Fig.\ref{T5N5} and  Fig.\ref{T1005N5} indicates that the zeros of the derivative of polynomials are on the unit circle. Thus we can strengthen Conjecture A to the claim that $T$-symmetrized polynomials $zq(z^T)$ might be quasi-extremal in the Genthner-Ruscheweyh-Salinas sense  (see Definition 11 and Theorem 12 in \cite{GRS}).

\section{Some particular cases}

\subsection{Case $T=1$}
In this case
$$
\psi_j=\frac{\pi(2j+1)}{N+1},\qquad \eta(z)=\sum_{j=1}^N\sin\frac{\pi j}{N+1}z^j,
$$
and we get the Suffridge polynomials
$$
F_T(z) = zq(z)=Kz\sum_{j=1}^N\left(1-\frac j{N+1}\right)\sin\frac{\pi j}{N+1}z^{j-1},\qquad
K=2\tan\frac\pi{2(N+1)}.
$$
Furthermore,
\(
P_N=\tan^2\frac\pi{2(N+1)}
\)
which agrees with \eqref{a1}.

\subsection{Case $T=2$}
In this case
$$
\psi_j=\pi\frac{2j}{N},\qquad \eta(z)=\sum_{j=1}^Nz^j,
$$
and we get the Fej\'er polynomial
$$
F_2(z) = zq(z)^2=K^2z\left(\sum_{j=1}^N\left(1-\frac{2j-1}{2N}\right)z^{j-1}\right)^2,\qquad
K=\frac1N.
$$
Furthermore,
$
(P_N)^2=\frac1{N^2},
$
which agrees with \eqref{a2}.\\

\section{The estimate for the product of cotangents}

The previous section provides the asymptotic behavior for $|I_N^{(T)}|$ in the cases $T=1,2.$ A natural question is whether we can find the asymptotic behavior for $P_N$ when $T\ge 3.$ The  goal of this section is to give an affirmative answer.  We rigorously prove the following result:  

\begin{theorem}\label{main}
 Let $T\geq 3.$ Then $P_N \approx N^{-\frac2T}$ as $N\to\infty$. More precisely, we get that
\begin{equation}\label{mainestodd}
 \lim_{N\to\infty} P_N N^{\frac{2}{T}} = \pi^{\frac{2-T}T}\Bigl(\Gamma(\frac{T+2}{2T})\Bigr)^2 
\end{equation}
\end{theorem}

Let us begin with the following very simple lemma which will turn out to be extremely useful:
 
\begin{lemma}\label{lem}
  For any $0<\gamma<\frac{\pi}2$ the function 
 \[
 f:[0,\frac{\pi}2 - \gamma)\to\mathbb R, \qquad f(x)=\frac{\tan x}{\tan(\gamma+x)}
 \]
 is increasing on $[0, \frac12(\frac{\pi}2-\gamma)]$ and decreasing on $[\frac12(\frac{\pi}2-\gamma), \frac{\pi}2 - \gamma)$
 \end{lemma}

\begin{proof}
 A simple computation gives that
\[
\frac{d}{dx} f(x) = \frac{\sin\gamma\cos(2x+\gamma)}{\sin^2(x+\gamma)\cos^2(x+\gamma)\cos^2 x}
\] 
and the conclusion of the lemma follows immediately.

\end{proof}

\begin{proof}

 We will start the proof of our main theorem with the case when $N$ is odd. We first note that, since $\cot (\frac{\pi}2 - x) = \tan x$, we can rewrite
\[
 P_N = \prod_{j=1}^{\frac{N-1}2} \frac{\tan \pi \frac{(2j-1)T}{2(2+(N-1)T)}}{\tan \pi \frac{2+(2j-1)T}{2(2+(N-1)T)}}
\] 

Let us now fix $m\ll N$, and define
\[
P_N^{(m,1)} = \prod_{j=m+1}^{\frac{N-1}2-m} \frac{\tan \pi \frac{(2j-1)T}{2(2+(N-1)T)}}{\tan \pi \frac{2+(2j-1)T}{2(2+(N-1)T)}}
\]

The strategy will be to first estimate $P_N^{(m,1)}$ from above and below, let $N\to\infty$, and then let $m\to\infty$. 

To accomplish the first task, let us first note that we can rewrite $P_N^{(m,1)}$ in the form
\[
P_N^{(m,1)} = \prod_{j=m+1}^{\lfloor\frac{N+1}4\rfloor-1} \frac{\tan \pi \frac{(2j-1)T}{2(2+(N-1)T)}}{\tan \pi \frac{2+(2j-1)T}{2(2+(N-1)T)}}\prod_{j=\lfloor\frac{N+1}4\rfloor}^{\lfloor\frac{N+1}4\rfloor+1} \frac{\tan \pi \frac{(2j-1)T}{2(2+(N-1)T)}}{\tan \pi \frac{2+(2j-1)T}{2(2+(N-1)T)}}\prod_{j=\lfloor\frac{N+1}4\rfloor+2}^{\frac{N-1}2-m} \frac{\tan \pi \frac{(2j-1)T}{2(2+(N-1)T)}}{\tan \pi \frac{2+(2j-1)T}{2(2+(N-1)T)}}
\]

Now let
\[
P_N^{(m, k)} = \prod_{j=m+1}^{\lfloor\frac{N+1}4\rfloor-1} \frac{\tan \pi \frac{2(k-1)+(2j-1)T}{2(2+(N-1)T)}}{\tan \pi \frac{2k + (2j-1)T}{2(2+(N-1)T)}}\prod_{j=\lfloor\frac{N+1}4\rfloor}^{\lfloor\frac{N+1}4\rfloor+1} \frac{\tan \pi \frac{(2j-1)T}{2(2+(N-1)T)}}{\tan \pi \frac{2+(2j-1)T}{2(2+(N-1)T)}}\prod_{j=\lfloor\frac{N+1}4\rfloor+2}^{\frac{N-1}2-m} \frac{\tan \pi \frac{2(1-k)+(2j-1)T}{2(2+(N-1)T)}}{\tan \pi \frac{2(2-k)+(2j-1)T}{2(2+(N-1)T)}}
\]
where $1\leq k \leq T$. Note that when $k=1$ we indeed recover the correct formula for $P_N^{(m,1)}$.

The first key observation is that
\begin{equation}\label{prod}
\prod_{k=1}^T P_N^{(m,k)} = \frac{\tan \pi\frac{(2m+1)T}{2(2+(N-1)T)}}{\tan \pi\frac{(2\lfloor\frac{N+1}4\rfloor-1)T}{2(2+(N-1)T)}}\Bigl(\prod_{j=\lfloor\frac{N+1}4\rfloor}^{\lfloor\frac{N+1}4\rfloor+1} \frac{\tan \pi \frac{(2j-1)T}{2(2+(N-1)T)}}{\tan \pi \frac{2+(2j-1)T}{2(2+(N-1)T)}}\Bigr)^T \, \frac{\tan \pi\frac{2+(2\lfloor\frac{N+1}4\rfloor+1)T}{2(2+(N-1)T)}}{\tan \pi \frac{2+(N-2m-2)T}{2(2+(N-1)T)}}
\end{equation}

The second key observation is that 
\begin{equation}\label{ineq}
P_N^{(m,k)} \leq P_N^{(m, k+1)}, \; 1\leq k\leq T-1.
\end{equation}

 Indeed, this is a simple consequence of Lemma~\ref{lem} applied for each term in the product with $\gamma = \pi\frac{2}{2(2+(N-1)T)}$.
 
As a consequence of \eqref{prod}, \eqref{ineq} and the fact that $\tan \pi\frac{(2m+1)T}{2(2+(N-1)T)} = \frac{1}{\tan \pi \frac{2+(N-2m-2)T}{2(2+(N-1)T)}}$ we obtain
\[
(P_N^{(m,1)})^T \leq \tan^2 \pi\frac{(2m+1)T}{2(2+(N-1)T)}\frac{\tan \pi\frac{2+(2\lfloor\frac{N+1}4\rfloor+1)T}{2(2+(N-1)T)}}{\tan \pi\frac{(2\lfloor\frac{N+1}4\rfloor-1)T}{2(2+(N-1)T)}} \Bigl(\prod_{j=\lfloor\frac{N+1}4\rfloor}^{\lfloor\frac{N+1}4\rfloor+1} \frac{\tan \pi \frac{(2j-1)T}{2(2+(N-1)T)}}{\tan \pi \frac{2+(2j-1)T}{2(2+(N-1)T)}}\Bigr)^T 
\]

Since 
\[
\lim_{N\to\infty} N^2 \tan^2 \pi\frac{(2m+1)T}{2(2+(N-1)T)}  = (\frac{\pi}2)^2\cdot (2m+1)^2
\]
and
\[
\lim_{N\to\infty} \frac{\tan \pi\frac{2+(2\lfloor\frac{N+1}4\rfloor+1)T}{2(2+(N-1)T)}}{\tan \pi\frac{(2\lfloor\frac{N+1}4\rfloor-1)T}{2(2+(N-1)T)}} \Bigl(\prod_{j=\lfloor\frac{N+1}4\rfloor}^{\lfloor\frac{N+1}4\rfloor+1} \frac{\tan \pi \frac{(2j-1)T}{2(2+(N-1)T)}}{\tan \pi \frac{2+(2j-1)T}{2(2+(N-1)T)}}\Bigr)^T = 1
\]
we obtain that
\[
\limsup_{N\to\infty} N^{\frac2T} P_N^{(m,1)} \leq \Bigl((\frac{\pi}2)^2\cdot (2m+1)^2\Bigr)^{\frac1T}
\]

 Moreover, since 
\[
P_N = \prod_{j=1}^m \Bigl(\frac{\tan \pi \frac{(2j-1)T}{2(2+(N-1)T)}}{\tan \pi \frac{2+(2j-1)T}{2(2+(N-1)T)}}\Bigr)^2 P_N^{(m,1)}
\] 
and 
\[
\lim_{N\to\infty} \prod_{j=1}^m \frac{\tan \pi \frac{(2j-1)T}{2(2+(N-1)T)}}{\tan \pi \frac{2+(2j-1)T}{2(2+(N-1)T)}} = \prod_{j=1}^m \frac{(2j-1)T}{2+(2j-1)T}
\]
we obtain that for any fixed $m$ 
\begin{equation}\label{sup}
\limsup_{N\to\infty} P_N N^{\frac{2}{T}} \leq (\frac{\pi}2)^{\frac2T} (2m+1)^{\frac2T}\prod_{j=1}^m \Bigl(\frac{(2j-1)T}{2+(2j-1)T}\Bigr)^2 := \alpha_m.
\end{equation}

 We use a similar argument to bound $\limsup_{N\to\infty} P_N N^{\frac{2}{T}}$. Let 
\[
Q_N^{(m,1)} = \prod_{j=m+1}^{\frac{N-1}2-m} \frac{\tan \pi \frac{(2j-1)T}{2(2+(N-1)T)}}{\tan \pi \frac{2+(2j-1)T}{2(2+(N-1)T)}}
\] 

Now, define $Q_N^{(m, k)}$ by the equation
\[
Q_N^{(m, k)}=\prod_{j=m+1}^{\lfloor\frac{N+1}4\rfloor} \frac{\tan \pi \frac{2(1-k)+(2j-1)T}{2(2+(N-1)T)}}{\tan \pi \frac{2(2-k) + (2j-1)T}{2(2+(N-1)T)}}\prod_{j=\lfloor\frac{N+1}4\rfloor+1}^{\frac{N-1}2-m} \frac{\tan \pi \frac{2(k-1)+(2j-1)T}{2(2+(N-1)T)}}{\tan \pi \frac{2k+(2j-1)T}{2(2+(N-1)T)}}
\]
where $1\leq k \leq T$. Note that when $k=1$ we indeed recover the correct formula for $Q_N^{(m,1)}$.

As above, one can show that
\[
Q_N^{(m,1)} \geq Q_N^{(k)}, \qquad 1\leq k\leq T.
\]
and
\[
 \prod_{k=1}^T Q_N^{(m,k)} = \frac{\tan \pi\frac{2+(2m-1)T}{2(2+(N-1)T)}}{\tan \pi\frac{2+(2\lfloor\frac{N+1}4\rfloor-1)T}{2(2+(N-1)T)}} \frac{\tan \pi\frac{(2\lfloor\frac{N+1}4\rfloor+1)T}{2(2+(N-1)T)}}{\tan \pi \frac{(N-2m)T}{2(2+(N-1)T)}}
\]

Since 
\[
\lim_{N\to\infty} N^2 \tan^2 \pi\frac{2+(2m-1)T}{2(2+(N-1)T)}  = (\frac{\pi}2)^2\cdot (2m-1)^2
\]
and
\[
\lim_{N\to\infty} \frac{\tan \pi\frac{(2\lfloor\frac{N+1}4\rfloor+1)T}{2(2+(N-1)T)}}{\tan \pi\frac{2+(2\lfloor\frac{N+1}4\rfloor-1)T}{2(2+(N-1)T)}} = 1
\]
we obtain that
\[
\liminf_{N\to\infty} N^{\frac2T} Q_N^{(m,1)} \geq \Bigl((\frac{\pi}2)^2\cdot (2m-1)^2\Bigr)^{\frac1T}
\]

 Moreover, since 
\[
P_N = \prod_{j=1}^m \Bigl(\frac{\tan \pi \frac{(2j-1)T}{2(2+(N-1)T)}}{\tan \pi \frac{2+(2j-1)T}{2(2+(N-1)T)}}\Bigr)^2 Q_N^{(m,1)}
\] 
and 
\[
\lim_{N\to\infty} \prod_{j=1}^m \frac{\tan \pi \frac{(2j-1)T}{2(2+(N-1)T)}}{\tan \pi \frac{2+(2j-1)T}{2(2+(N-1)T)}} = \prod_{j=1}^m \frac{(2j-1)T}{2+(2j-1)T}
\]
we obtain that for any fixed $m$ 
\begin{equation}\label{inf}
\liminf_{N\to\infty} P_N N^{\frac{2}{T}} \geq (\frac{\pi}2)^{\frac2T} (2m-1)^{\frac2T}\prod_{j=1}^m \Bigl(\frac{(2j-1)T}{2+(2j-1)T}\Bigr)^2 := \beta_m.
\end{equation}

From \eqref{sup} and \eqref{inf} we obtain that
\[
\beta_m \leq \liminf_{N\to\infty} P_N N^{\frac{2}{T}} \leq \limsup_{N\to\infty} P_N N^{\frac{2}{T}} \leq \alpha_m
\]

It is now easy to check that $\beta_m$ is an increasing sequence, $\alpha_m$ is a decreasing sequence, and $\lim_{m\to\infty}\frac{\alpha_m}{\beta_m}=1$. Thus we obtain that $\lim_{N\to\infty} P_N N^{\frac{2}{T}}$ exists and equals $\lim_{m\to\infty} \alpha_m$. To finish the case of $N$ odd, we are left to prove that
\begin{equation}\label{limNodd}
\lim_{m\to\infty} \alpha_m = \pi^{\frac{2-T}T}\Bigl(\Gamma(\frac{T+2}{2T})\Bigr)^2
\end{equation}

We can write
\[
\alpha_m = \pi^{\frac2T} (m+\frac12)^{\frac2T} \Bigl(\frac{\prod_{j=1}^m (2jT - T)}{\prod_{j=1}^m (2jT - (T-2))}\Bigr)^2
\]

Since $\Gamma(x) = x \Gamma(x-1)$, one easily obtains:
\[
\prod_{j=1}^m (jr-s) = r^m \Gamma(m+\frac{r-s}r)/\Gamma(\frac{r-s}r), \qquad 0\leq s\leq r-1
\]

By applying the above formula with $r=2T$, $s=T$ for the numerator, and $r=2T$, $s=T-2$ for the denominator we obtain
\[
\alpha_m = \pi^{\frac2T} (m+\frac12)^{\frac2T} \Bigl(\frac{\Gamma(m+\frac12)}{\Gamma(m+\frac{T+2}{2T})}\frac{\Gamma(\frac{T+2}{2T})}{\Gamma(\frac12)}\Bigr)^2
\]

We now use the asymptotic for the $\Gamma$ function
\[
\Gamma(z) = e^{-z}z^{z-1/2}\sqrt{2\pi}(1+O(1/z))
\]
to obtain
\[\begin{split}
\frac{\Gamma(m+\frac12)}{\Gamma(m+\frac{T+2}{2T})} & = \frac{e^{-(m+\frac12)}(m+\frac12)^{(m+\frac12)-1/2}\sqrt{2\pi}(1+O(1/m))}{e^{-(m+\frac{T+2}{2T})}(m+\frac{T+2}{2T})^{(m+\frac{T+2}{2T})-1/2}\sqrt{2\pi}(1+O(1/m))} \\ & = e^{\frac1T} (m+\frac12)^{-\frac1T}\Bigl(\frac{m+\frac12}{m+\frac{T+2}{2T}}\Bigr)^{m+\frac1T} + O(1/m) = (m+\frac12)^{-\frac1T} + O(1/m)
\end{split}\]
since
\[
\lim_{m\to\infty}\Bigl(\frac{m+\frac12}{m+\frac{T+2}{2T}}\Bigr)^{m+\frac1T} = \lim_{m\to\infty}\Bigl(1 -\frac{\frac1T}{m+\frac{T+2}{2T}}\Bigr)^{m+\frac1T} = e^{-\frac1T}
\]

We thus obtain 
\[
\lim_{m\to\infty} \alpha_m = \pi^{\frac2T}\Bigl(\frac{\Gamma(\frac{T+2}{2T})}{\Gamma(\frac12)}\Bigr)^2
\]
and \eqref{limNodd} follows since $\Gamma(\frac12) = \sqrt{\pi}$.

 The proof for $N$ even follows in a similar manner. As above, we can rewrite
\[
P_N = \frac{T}{2+(N-1)T} \tP_N
\] 
where
\[
\tP_N = \prod_{j=1}^{\frac{N-2}2} \frac{\tan \pi \frac{2jT}{2(2+(N-1)T)}}{\tan \pi \frac{2+(2j-1)T}{2(2+(N-1)T)}}
\]

 For any $m\ll N$, let
\[
\tP_N^{(m,1)} = \prod_{j=m+1}^{\frac{N-2}2-m} \frac{\tan \pi \frac{2jT}{2(2+(N-1)T)}}{\tan \pi \frac{2+(2j-1)T}{2(2+(N-1)T)}}
\] 
and define
\[
\tP_N^{(m, k)} = \prod_{j=m+1}^{\lfloor\frac{N}4\rfloor} \frac{\tan \pi \frac{-(k-1)+2jT}{2(2+(N-1)T)}}{\tan \pi \frac{-(k-1) + 2+(2j-1)T}{2(2+(N-1)T)}} \prod_{j=\lfloor\frac{N}4\rfloor+1}^{\frac{N-2}2-m} \frac{\tan \pi \frac{(k-1)+2jT}{2(2+(N-1)T)}}{\tan \pi \frac{(k-1) + 2+(2j-1)T}{2(2+(N-1)T)}}
\]
\[
\tQ_N^{(m, k)} = \prod_{j=m+1}^{\lfloor\frac{N}4\rfloor-1} \frac{\tan \pi \frac{(k-1)+2jT}{2(2+(N-1)T)}}{\tan \pi \frac{(k-1) + 2+(2j-1)T}{2(2+(N-1)T)}} \prod_{j=\lfloor\frac{N}4\rfloor}^{\lfloor\frac{N}4\rfloor+1} \frac{\tan \pi \frac{2jT}{2(2+(N-1)T)}}{\tan \pi \frac{2+(2j-1)T}{2(2+(N-1)T)}} \prod_{j=\lfloor\frac{N}4\rfloor+1}^{\frac{N-2}2-m} \frac{\tan \pi \frac{(k-1)+2jT}{2(2+(N-1)T)}}{\tan \pi \frac{(k-1) + 2+(2j-1)T}{2(2+(N-1)T)}}
\]
where $1\leq k \leq 2T$. We observe that $\tP_N^{(m, 1)} = \tQ_N^{(m, 1)}$ and by applying Lemma~\ref{lem} as in the odd case,
\[
\tQ_N^{(m, k)} \leq \tP_N^{(m, 1)} \leq \tP_N^{(m, k)}, \qquad 1\leq k\leq 2T
\]

 Moreover, we have
\[
\prod_{k=1}^{2T} \tP_N^{(m,k)} = \prod_{i=0}^{T-3}\frac{1}{\tan^2 \pi \frac{-i-T+(2m+1)T}{2(2+(N-1)T)}} F(N)
\]
\[
\prod_{k=1}^{2T} \tQ_N^{(m,k)} = \prod_{i=0}^{T-3}\frac{1}{\tan^2 \pi \frac{i+2+(2m+1)T}{2(2+(N-1)T)}} G(N)
\]
with $\lim_{N\to\infty} F(N)=\lim_{N\to\infty} G(N) = 1$.

Since 
\[\begin{split}
\lim_{N\to\infty} N^{2(2-T)} \prod_{i=0}^{T-3}\frac{1}{\tan^2 \pi \frac{-i-T+(2m+1)T}{2(2+(N-1)T)}} F(N)  = (\frac2{\pi})^{2(T-2)}\prod_{i=0}^{T-3} \frac{T^2}{(2mT-i)^2}
\end{split}\]
\[\begin{split}
\lim_{N\to\infty} N^{2(2-T)} \prod_{i=0}^{T-3}\frac{1}{\tan^2 \pi \frac{i+2+(2m+1)T}{2(2+(N-1)T)}} G(N)  = (\frac2{\pi})^{2(T-2)}\prod_{i=0}^{T-3} \frac{T^2}{(2mT+T+i+2)^2}
\end{split}\]
we obtain that 
\[
\limsup_{N\to\infty} N^{\frac{2-T}T} \tP_N^{(m,1)} \leq \Big[(\frac2{\pi})^{2(T-2)}\prod_{i=0}^{T-3} \frac{T^2}{(2mT-i)^2}\Bigr]^{\frac1{2T}}
\]
\[
\liminf_{N\to\infty} N^{\frac{2-T}T} \tP_N^{(m,1)} \geq \Big[(\frac2{\pi})^{2(T-2)}\prod_{i=0}^{T-3} \frac{T^2}{(2mT+T+i+2)^2}\Bigr]^{\frac1{2T}}
\]

 Moreover, since 
\[
P_N = \frac{T}{2+(N-1)T} \prod_{j=1}^m \Bigl(\frac{\tan \pi \frac{2jT}{2(2+(N-1)T)}}{\tan \pi \frac{2+(2j-1)T}{2(2+(N-1)T)}}\Bigr)^2 \tP_N^{(m,1)}
\] 
we obtain for every $m$ that
\[
\tilde\beta_m \leq \liminf_{N\to\infty} P_N N^{\frac{2}{T}} \leq \limsup_{N\to\infty} P_N N^{\frac{2}{T}} \leq \tilde\alpha_m
\]
with
\[
\tilde\beta_m = \prod_{j=1}^m \Bigl(\frac{2jT}{2+(2j-1)T}\Bigr)^2 \pi^{\frac{2-T}T} \prod_{i=0}^{T-3} \Bigl[\frac{T}{mT+(T+i+2)/2}\Bigr]^{\frac1{T}}
\]
\[
\tilde\alpha_m = \prod_{j=1}^m \Bigl(\frac{2jT}{2+(2j-1)T}\Bigr)^2 \pi^{\frac{2-T}T} \prod_{i=0}^{T-3} \Bigl[\frac{T}{mT-i/2}\Bigr]^{\frac1{T}}
\]

 Since 
\[
\lim_{m\to\infty} m^{\frac{T-2}T} \prod_{i=0}^{T-3} \Bigl[\frac{T}{mT+(T+i+2)/2}\Bigr]^{\frac1{T}} = \lim_{m\to\infty} m^{\frac{T-2}T} \prod_{i=0}^{T-3} \Bigl[\frac{T}{mT-i/2}\Bigr]^{\frac1{T}} = 1
\] 

it follows that $\lim_{m\to\infty} \frac{\tilde\alpha_m}{\tilde\beta_m} = 1$, and it is enough to prove that
\[
\lim_{m\to\infty} m^{\frac{2-T}T} \prod_{j=1}^m \Bigl(\frac{2jT}{2+(2j-1)T}\Bigr)^2 \pi^{\frac{2-T}T} = \pi^{\frac{2-T}T}\Bigl(\Gamma(\frac{T+2}T)\Bigr)^2
\]

 But this follows just like in the case of $N$ odd by using
\[
\prod_{j=1}^m 2jT = (2T)^m \Gamma(m+1)
\] 
\[
\prod_{j=1}^m (2+(2j-1)T) = (2T)^m \frac{\Gamma(m+\frac{T+2}{2T})}{\Gamma(\frac{T+2}{2T})}
\]
\[
\lim_{m\to\infty} \frac{\Gamma(m+1)}{\Gamma(m+\frac{T+2}{2T})} m^{\frac{2-T}{2T}} = 1
\]

\end{proof}

Note that if $T=1$ then $\Gamma(3/2)=\sqrt\pi/2$ and if $T=2$ then $\Gamma(2)=1$. Substituting in \eqref{mainestodd} yields $\pi^2/(4N^2)$ for $T=1$, which is the asymptotic behavior in \eqref{I1N}, and exactly $1/N^2$ which is the value in \eqref{I2N}. This justifies the following \\

{\bf Conjecture C}: {\it 
Let $T\geq 3$ and $N\ge 1$ be integers. Then $|I_N^{(T)}|= (P_N)^T.$\\
}\\

The proven theorem gives a reasonable approximation to the interval $(-(P_N)^T,1)$ that is conjectured to be the optimal range of multipliers $(-\mu^*,1).$ It is remarkable that asymptotically the range for the admissible multipliers grows as $N^{2}$, regardless of the value of $T.$ The dependence on $T$ is in the constant.

It is also interesting to compare the left and right sides of the relation \eqref{mainestodd}. Numeric simulations indicates that they are pretty close even for small values of $T$ and $N.$
Say, for $T=5$ and $N=5$ we have 
$
F_T(-1)=0.02211102001\approx0.01752033601.
$
For T=105 and N=55, 
$
F_T(-1)=0.00006979604353\approx0.00006734173127.
$
For T=1005 and N=25,
$
F_T(-1)=0.0003388694786\approx0.0003126559570.
$
For T=1005 and N=35,
$
F_T(-1)=0.0001689295955\approx0.0001595183454.
$
For T=1005 and N=55,
$
F_T(-1)=0.00006699551479\approx0.00006459833822.
$

\section{Applications to non-linear dynamics}

 Scalar discrete systems always have real multipliers, so our method can be applied. We provided an example of finding an 8-cycle for the logistic map in \cite{DKST}. In this paper we choose to provide several examples of cycle detection in the vector case. Since we do not know whether the multipliers are real negative there is no guarantee that applying averaging we will find a cycle. Thus, we simply apply the averaging procedure and check if we found the cycle of the given length. We first run several hundreds iterations of the open loop system, and then we switch to averaging using the produced chaotic orbits as  initial values.\\

\subsection{H\'enon Map} The H\'enon Map is defined by
$$
\begin{cases}
x_{n+1}=&1-ax_n^2+by_n,\\
y_{n+1}=&x_n.
\end{cases}
\qquad  |b_n|<1, a_n>0.
$$
\begin{figure}[!htbp]
\centering
\begin{minipage}[b]{0.75\linewidth}
\centerline{\includegraphics[scale=0.55]{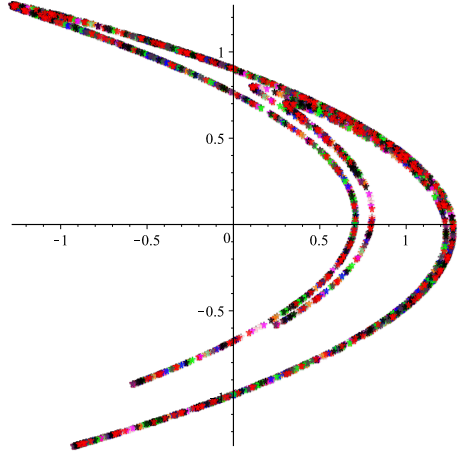}}
\vspace{-6cm}
\caption{H\'enon map}
\label{henon}
\end{minipage}
\end{figure}

Fig.\ref{henon} displays the strange attractor of the map, with different colors corresponding to orbits with the different initial values.

The  H\'enon map was well studied by mathematicians (c.f. for example \cite{BC}). It is standard to pick $a=1.4, b=0.3.$, in which case there are two equilibria:  $(x_1,y_1)=(0.631354477,0.631354477)$ and $(x_2,y_2)=(-1.3199135566,-1.3199135566)$ with the corresponding multipliers
$$
\mu_{1,2}(x_1,y_1)\in \{ 0.1559463223, -1.923738858\} 
$$
and 
$$
\mu_{1,2}(x_2,y_2)\in \{-0.09202956204, 3.259822098\}
$$
We remark that the second equilibrium has a positive multiplier greater than 1, in which case our averaging method does not work.\\

We have applied 400 iterations to the initial system \eqref{system} and then the system \eqref{system} was replaced with the averaging system \eqref{ave} for $T=1, N=8.$ Fig.~\ref{henon1} displays the orbits of system \eqref{ave} for $n=590,...,690$ iterations. 
 
\begin{figure}[!htbp]
\centering
\begin{minipage}[b]{0.75\linewidth}
\centerline{\includegraphics[scale=0.35]{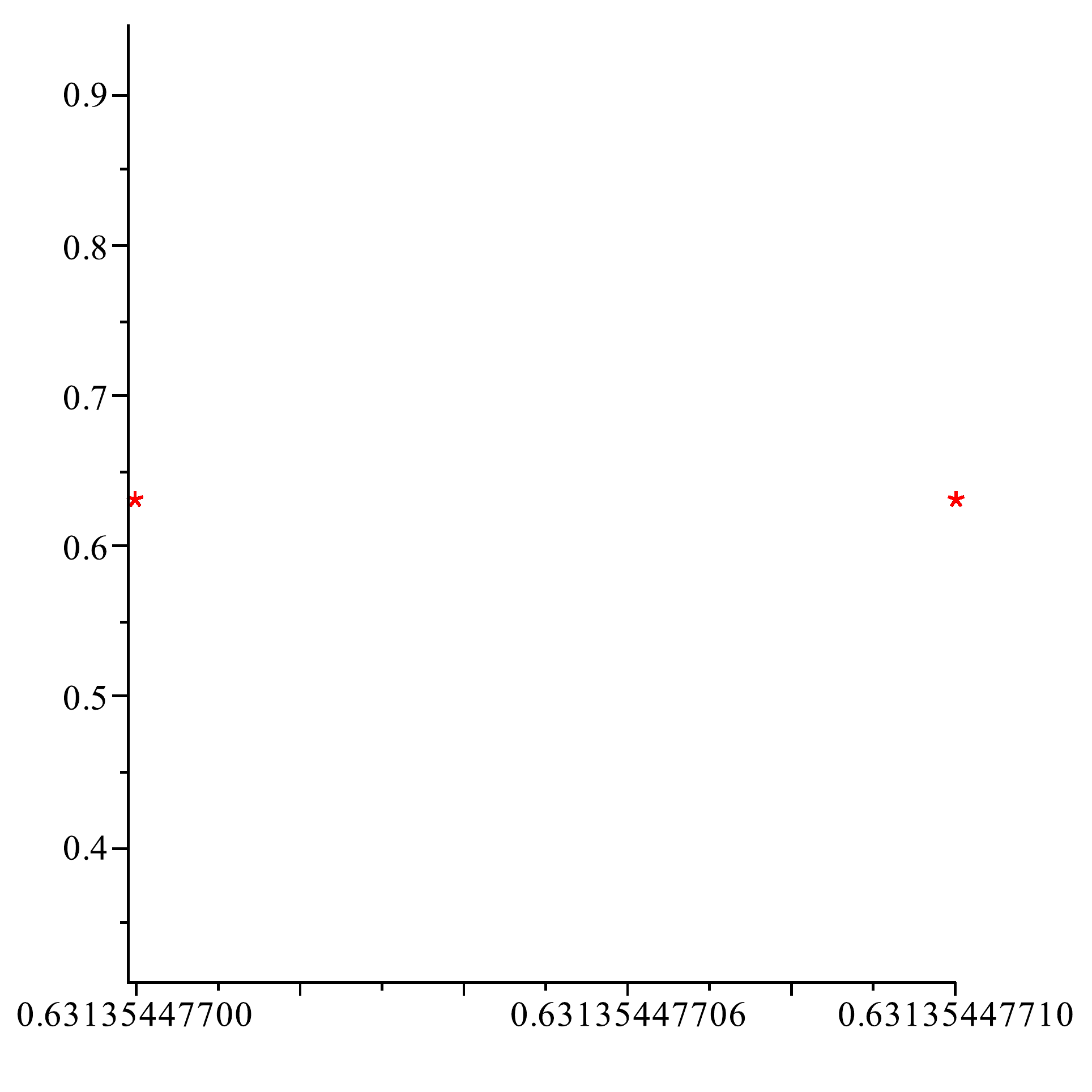}\hspace{0.2cm}
\includegraphics[scale=0.35]{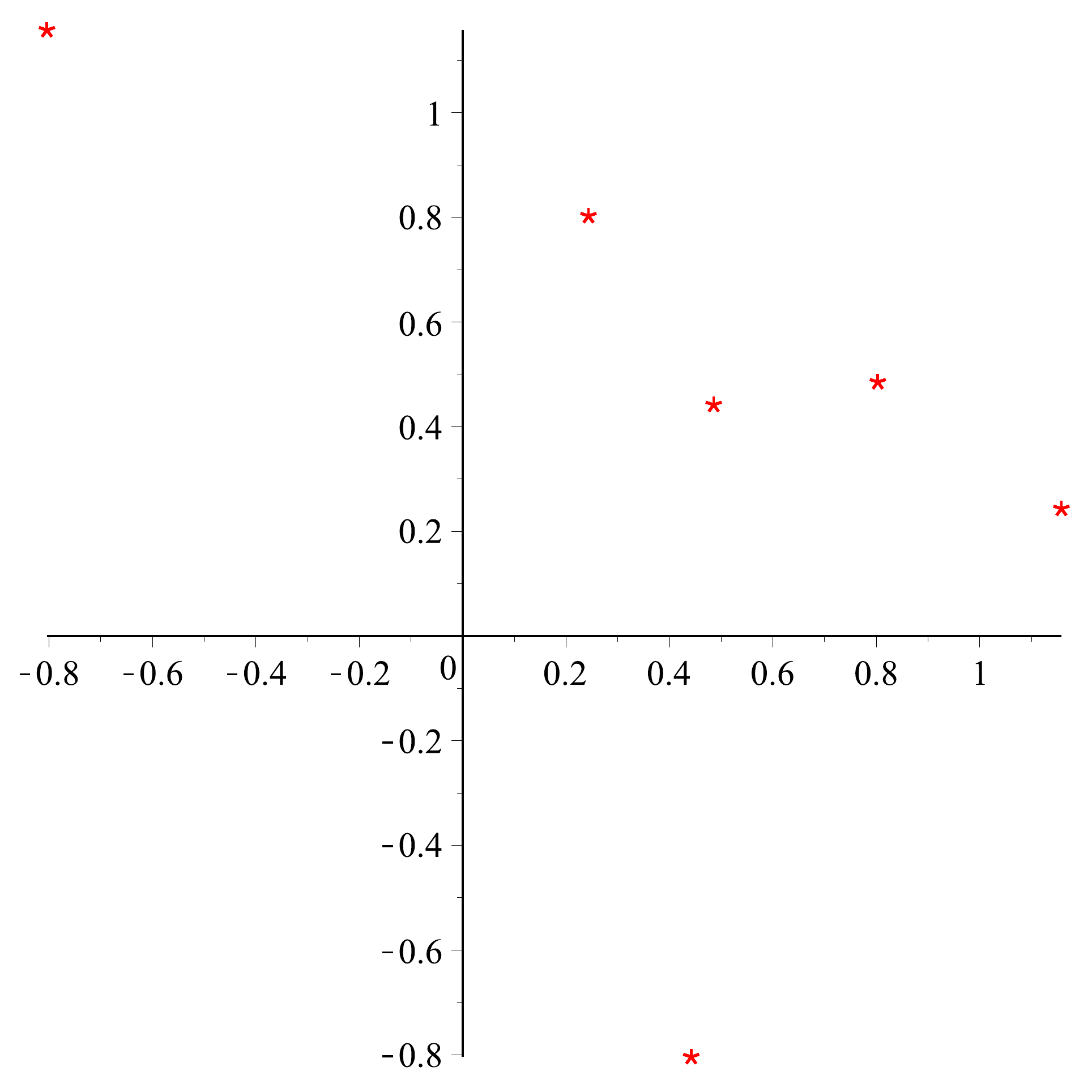}}
\caption{H\'enon Map equilibrium and 6-cycle}
\label{henon1}
\end{minipage}
\end{figure}
One can observe only two spots, which is actually only one since the difference is very small. The spots are red because in MAPLE code the red color was assigned to the last initial value, but all the other colors  converge to this equilibrium. Fig.\ref{henon1} also displays a 6-cycle of the map.\\

\subsection{Lozi Map} 
The Lozi Map is a two-dimensional map similar to the H\'enon map but with a different nonlinear term. It is given by the equations

$$
\begin{cases}
x_{n+1}=&1-\alpha|x_n|+y_n	\\
y_{n+1}=&\beta x_n.	
\end{cases}
$$

The standard choice of parameters is $\alpha=1.4, \beta=0.3.$

\begin{figure}[!htbp]
\centering
\begin{minipage}[b]{0.75\linewidth}
\centerline{\includegraphics[scale=0.3]{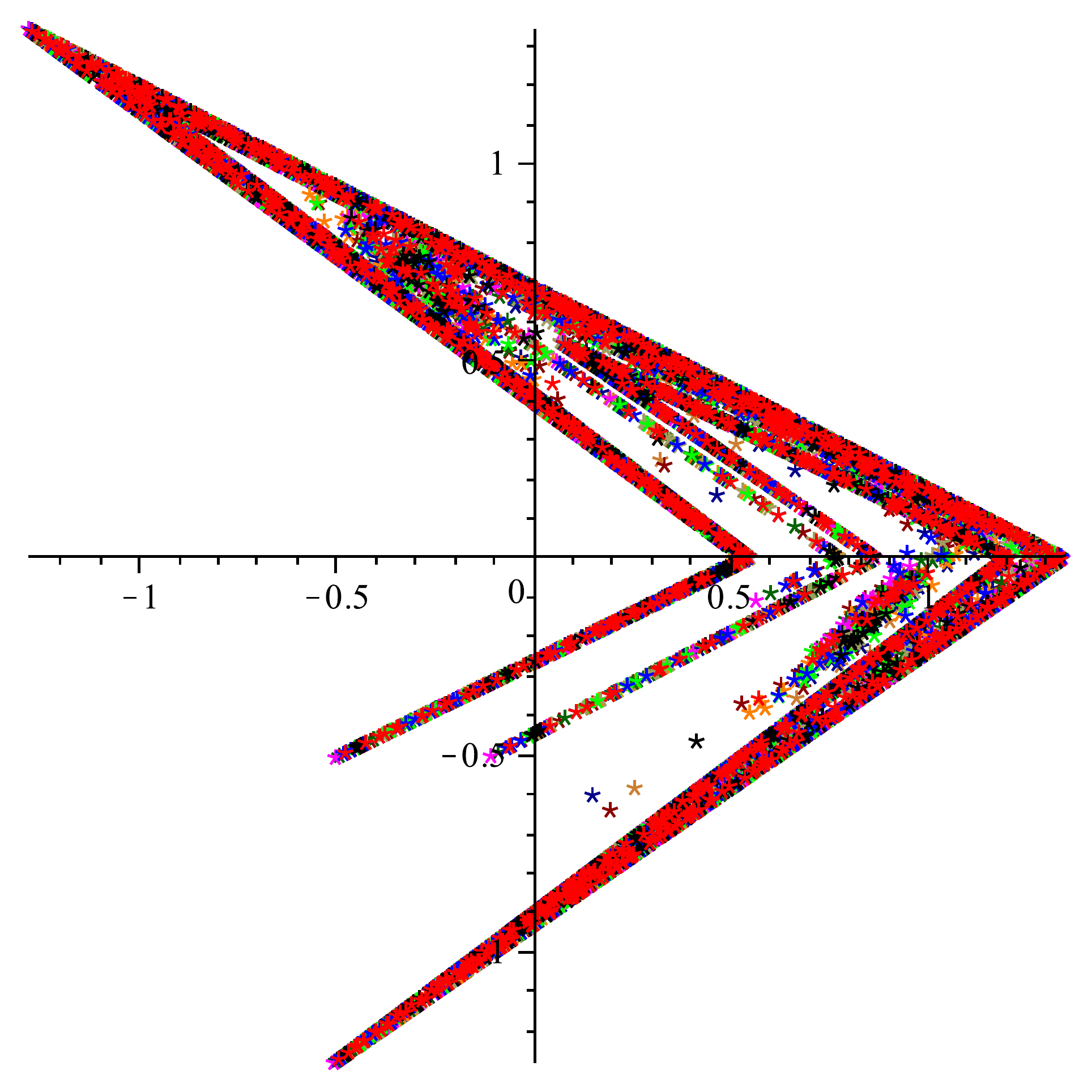}}
\caption{Lozi map}
\label{Lozi}
\end{minipage}
\end{figure}

Fig.\ref{Lozi} displays the strange attractor of the map. Fig.\ref{lozi4} displays the 6 and  8 cycles in the Lozi map found by the averaging method with $N=6.$  

\begin{figure}[!htbp]
\centering
\begin{minipage}[b]{0.75\linewidth}
\centerline{\includegraphics[scale=0.35]{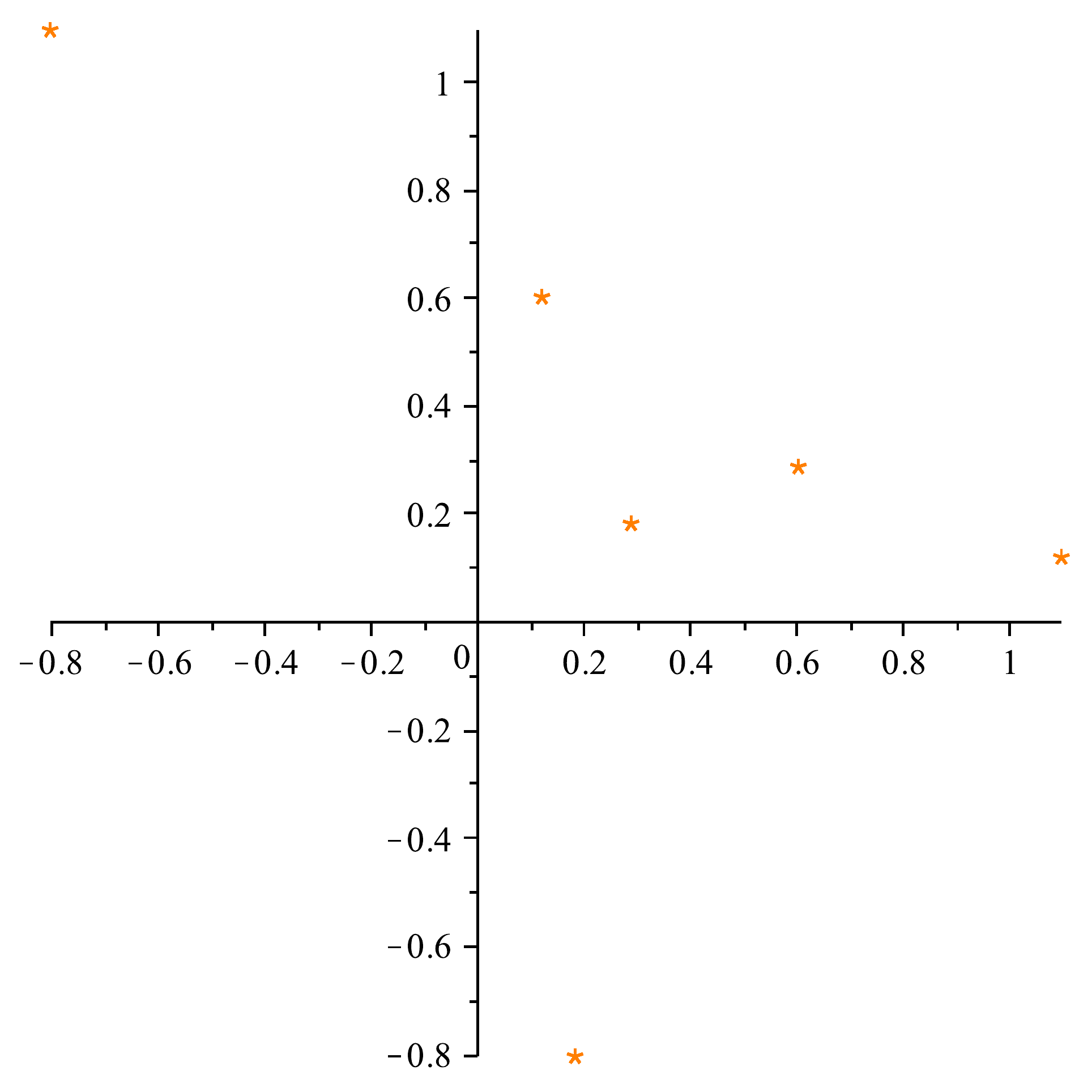}
\includegraphics[scale=0.35]{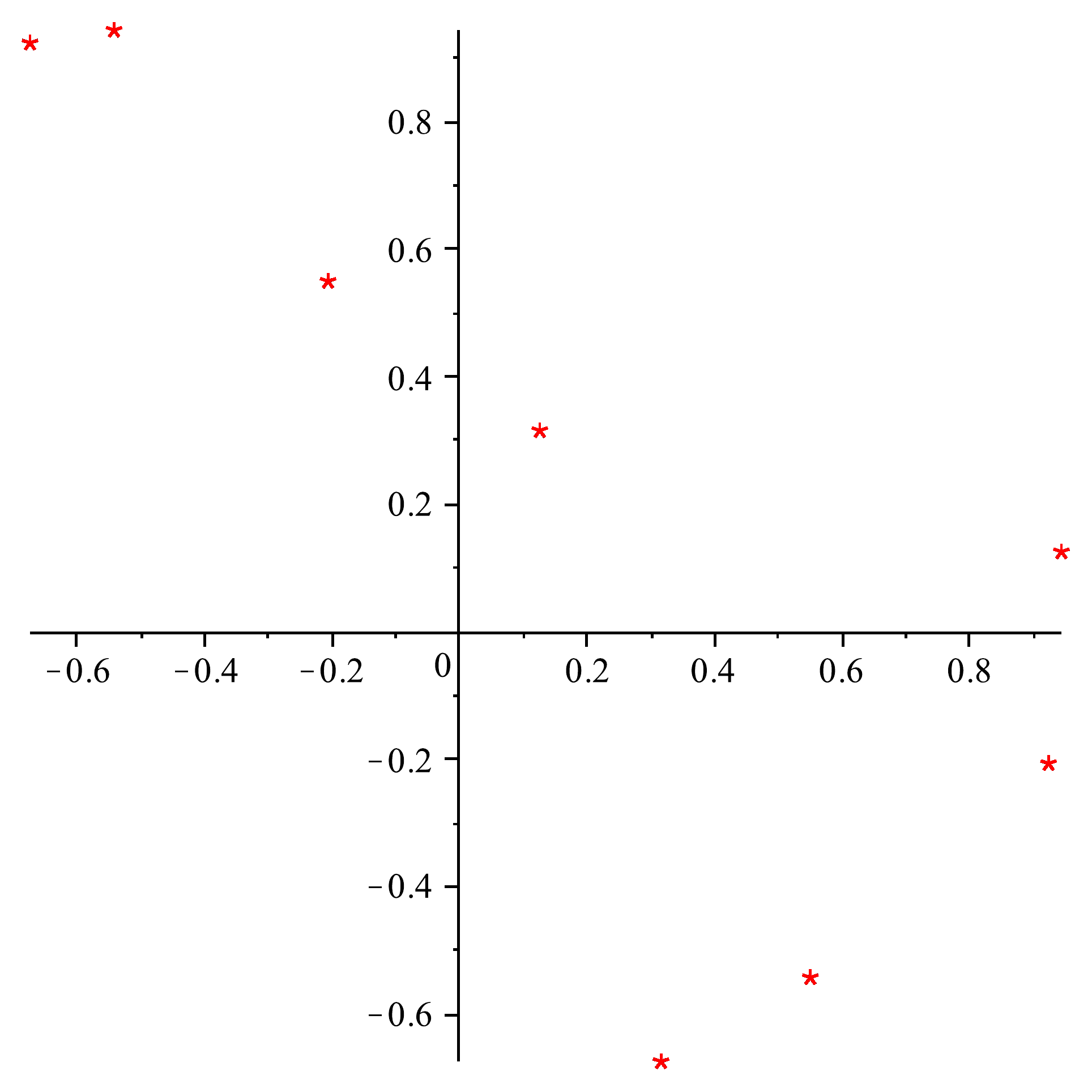}
}
\caption{6 and 8 cycles in the Lozi map}
\label{lozi4}
\end{minipage}
\end{figure}

\newpage

\subsection{Ikeda map}  The Ikeda map is given by
$$
\begin{cases}
x_{n+1}=&1+0.9\left(x_n\cos\left(0.4-\frac6{1+x_n^2+y_n^2}\right) - y_n \sin\left(0.4-\frac6{1+x_n^2+y_n^2}\right) \right),
\\
y_{n+1}=& 0.9\left(x_n\sin\left(0.4-\frac6{1+x_n^2+y_n^2}\right) + y_n \cos\left(0.4-\frac6{1+x_n^2+y_n^2}\right) \right),
\end{cases}
$$
Figure \ref{ikeda5} displays the initial Ikeda map and the 5-cycle found by averaging the initial system starting from n=400 to n=7800 and $N=6.$ 

\begin{figure}[!htbp]
\centering
\begin{minipage}[b]{0.75\linewidth}
\centerline{\includegraphics[scale=0.35]{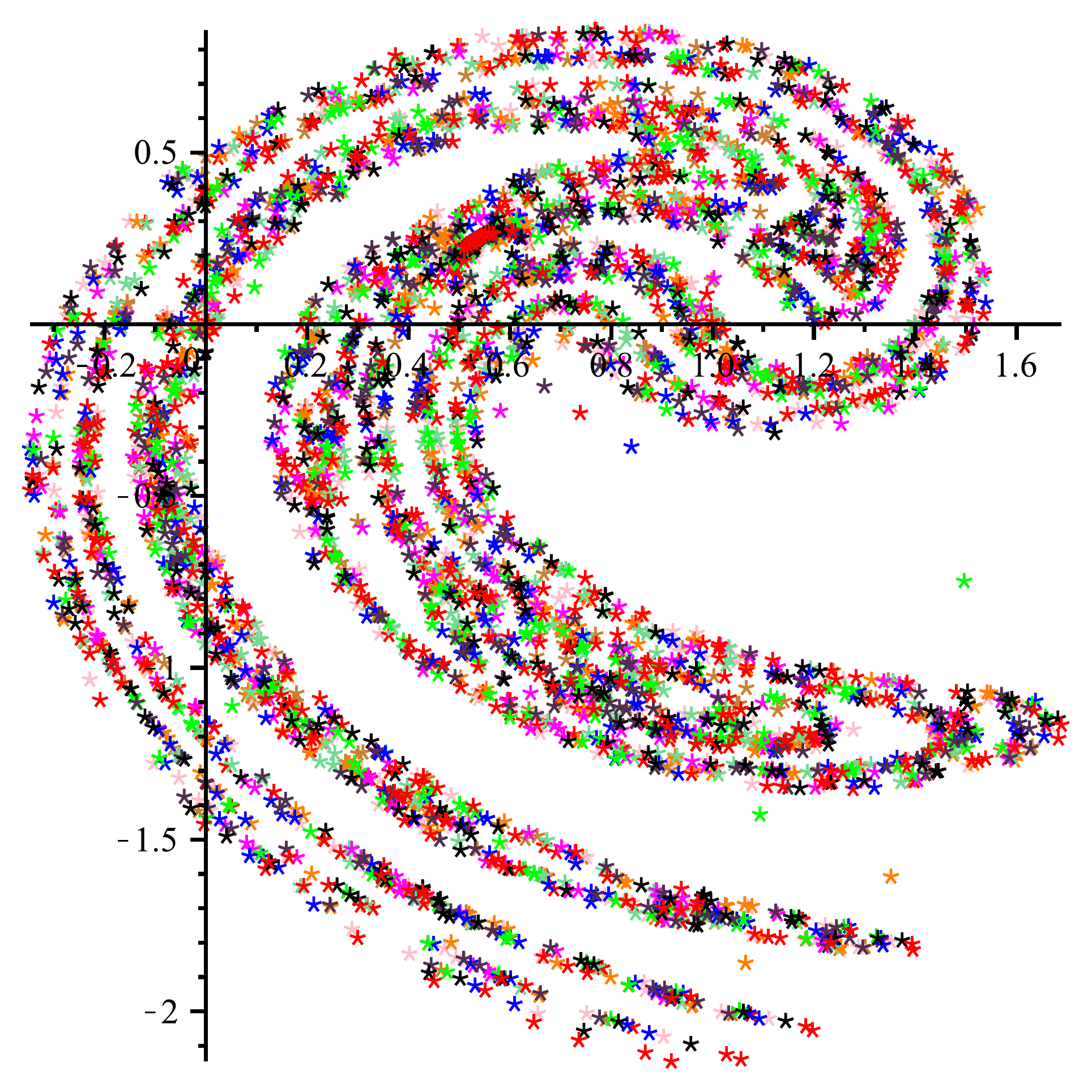}
\includegraphics[scale=0.35]{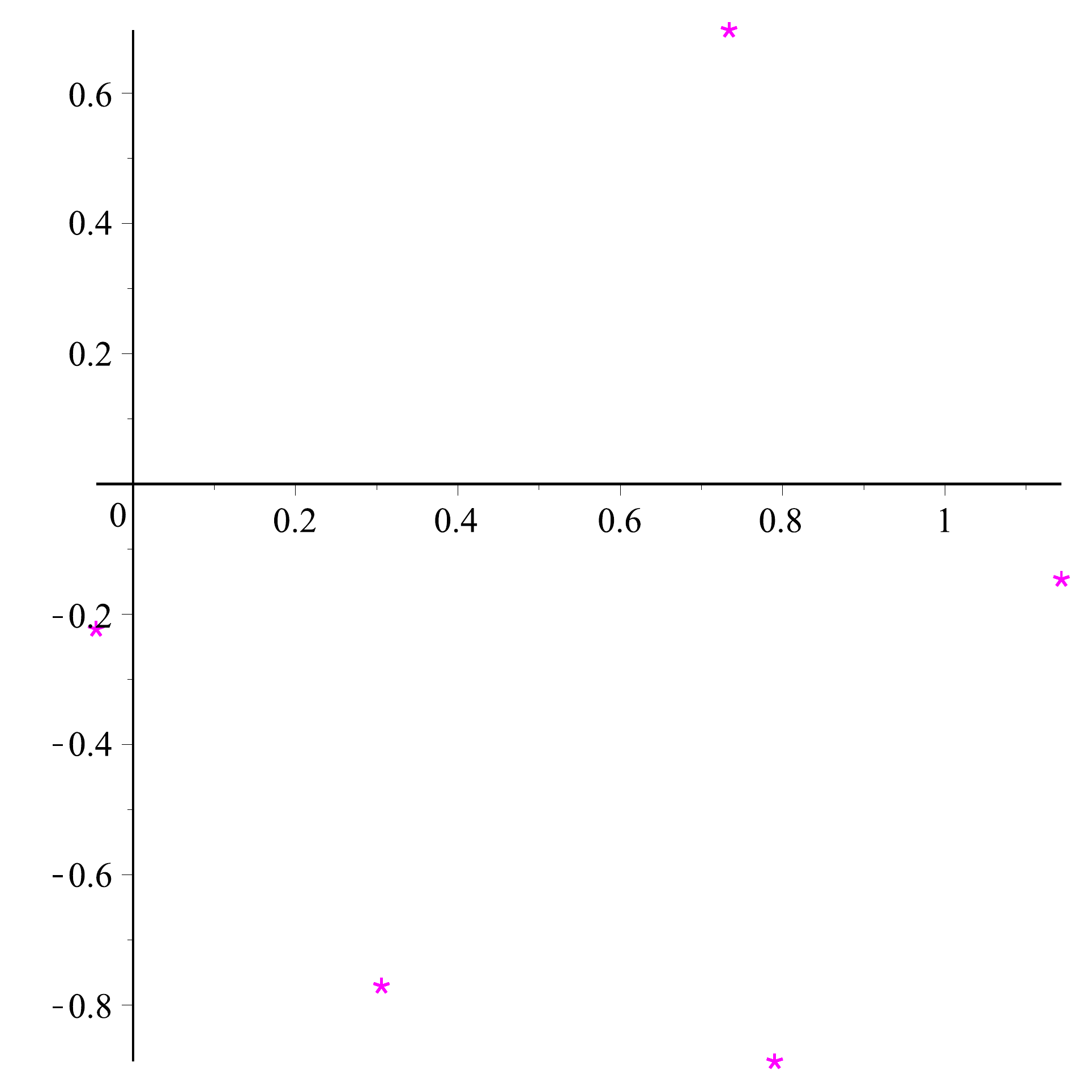}}
\caption{5-cycle in the Ikeda Map}
\label{ikeda5}
\end{minipage}
\end{figure}


\subsection{Elhaj-Sprott map}

The Elhaj-Sprott map is given by the system
$$
\begin{cases}
x_{n+1}=&1-4\sin(x_n)+0.9y_n	\\
y_{n+1}=& x_n.	
\end{cases}
$$
After averaging with  parameters $N= 2$, $T=4$ we find a 4-cycle after 3900 iterations
\begin{figure}[!htbp]
\centering
\begin{minipage}[b]{0.75\linewidth}
\centerline{\includegraphics[scale=0.35]{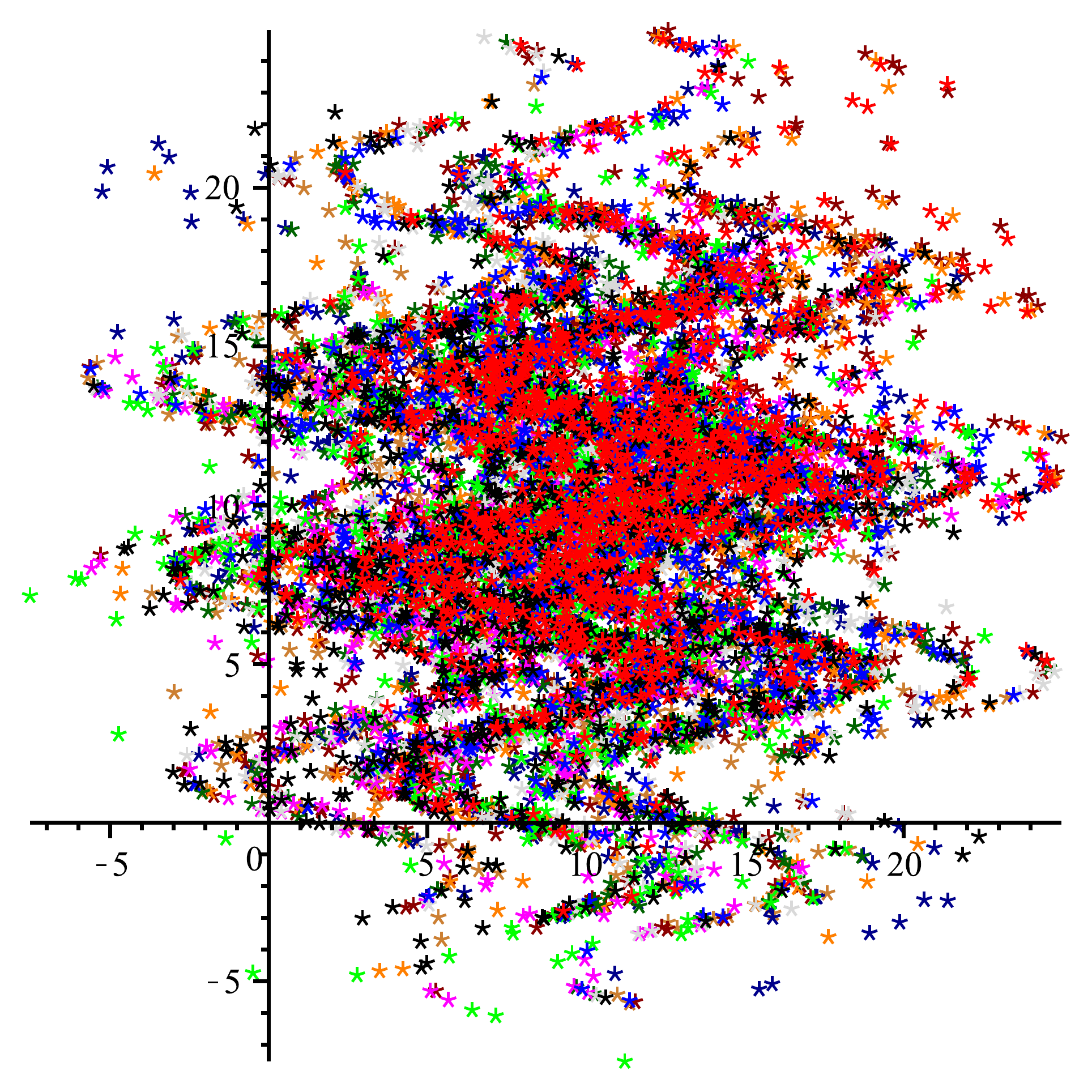}
\includegraphics[scale=0.35]{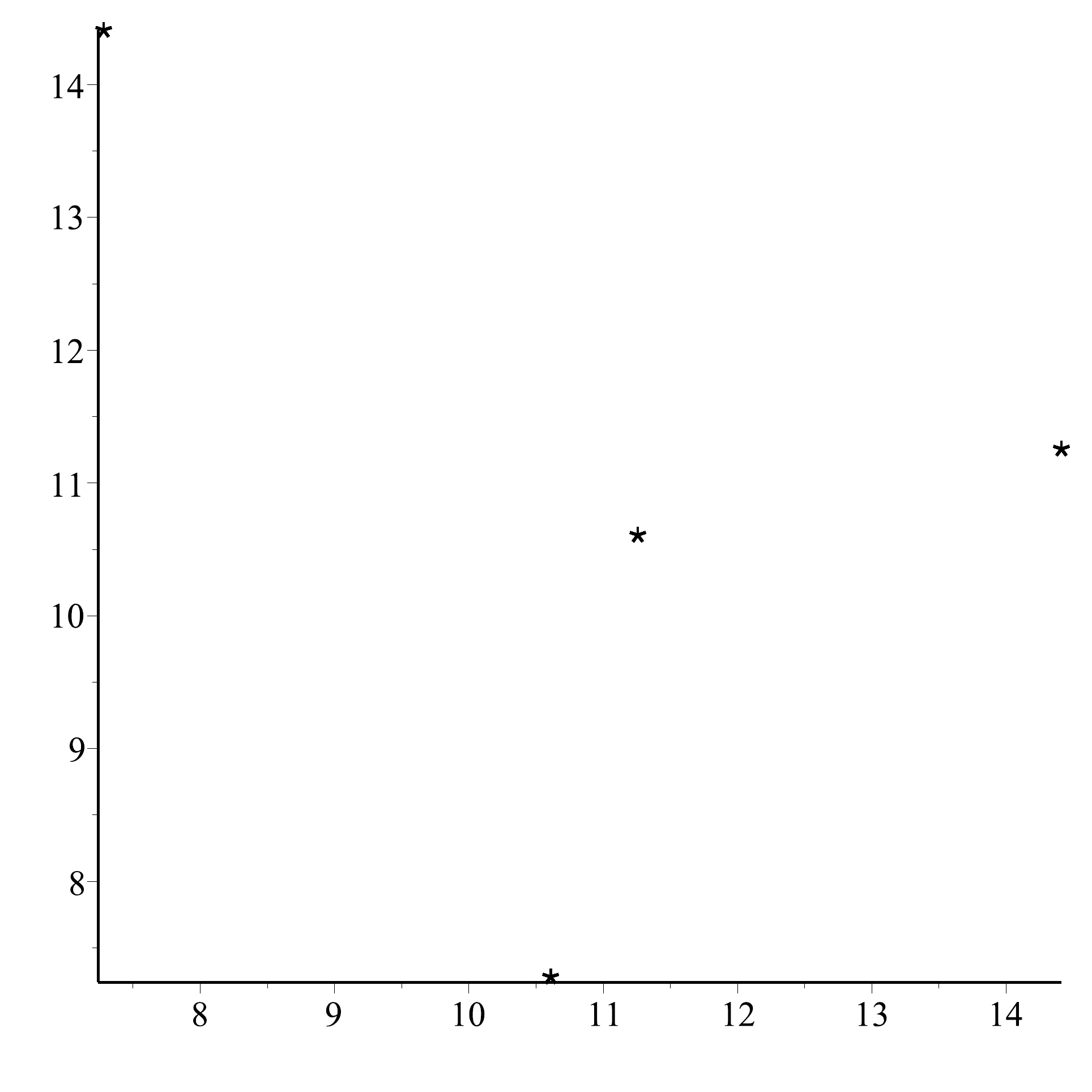}}
\caption{4-cycle in the Elhaj-Sprott map}
\label{elsp}
\end{minipage}
\end{figure}


\subsection{Holmes cubic map}
The Holmes cubic map is given by the equations
$$
\begin{cases}
x_{n+1}=&y_n	\\
y_{n+1}=&0.2x_n+2.77y_n-y_n^3.	
\end{cases}
$$
Below we average with parameters $N= 2, T=2.$ The 2-cycle is detected after 5700 iterations

\begin{figure}[!htbp]
\centering
\begin{minipage}[b]{0.75\linewidth}
\centerline{\includegraphics[scale=0.65]{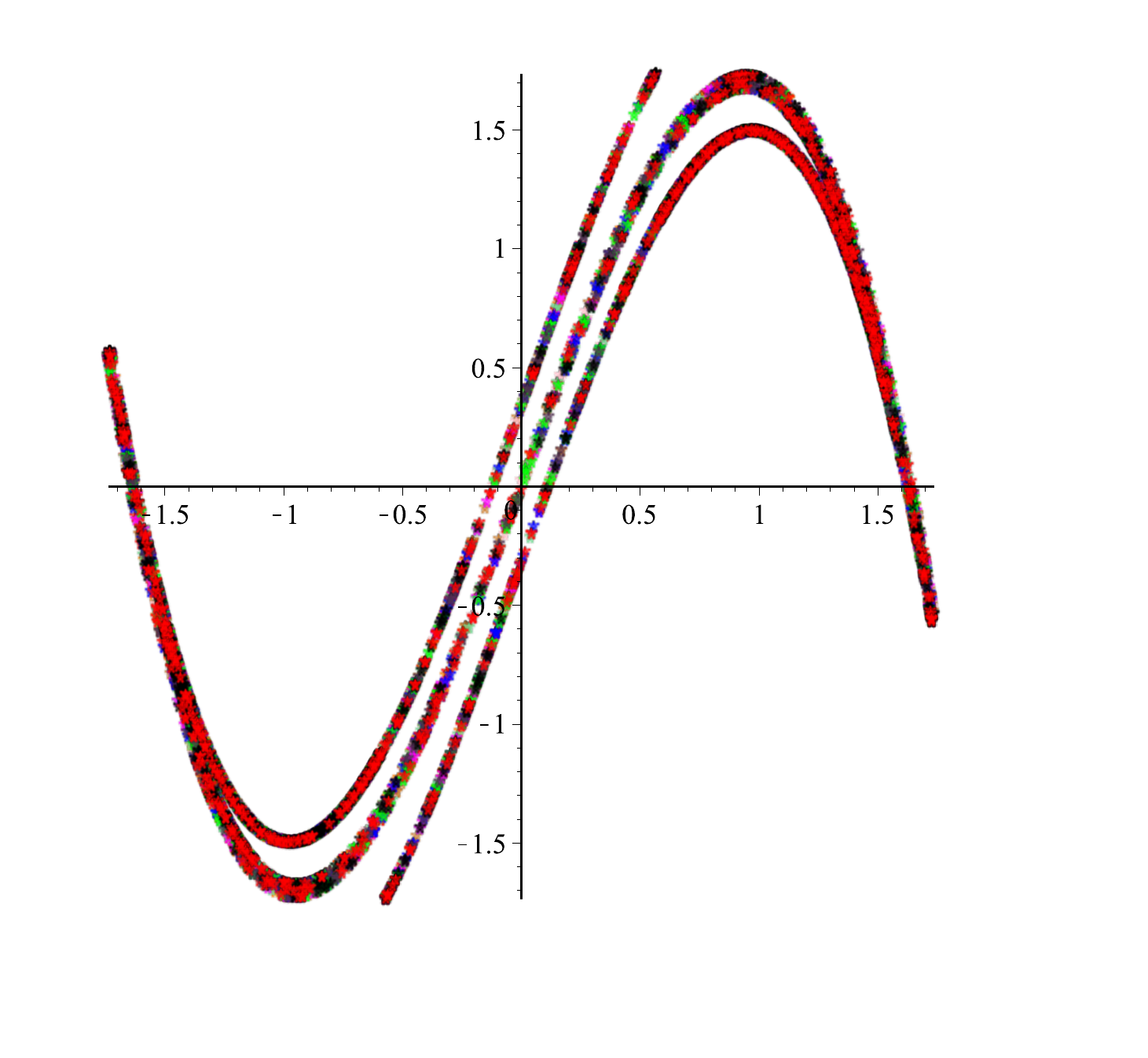}
\includegraphics[scale=0.35]{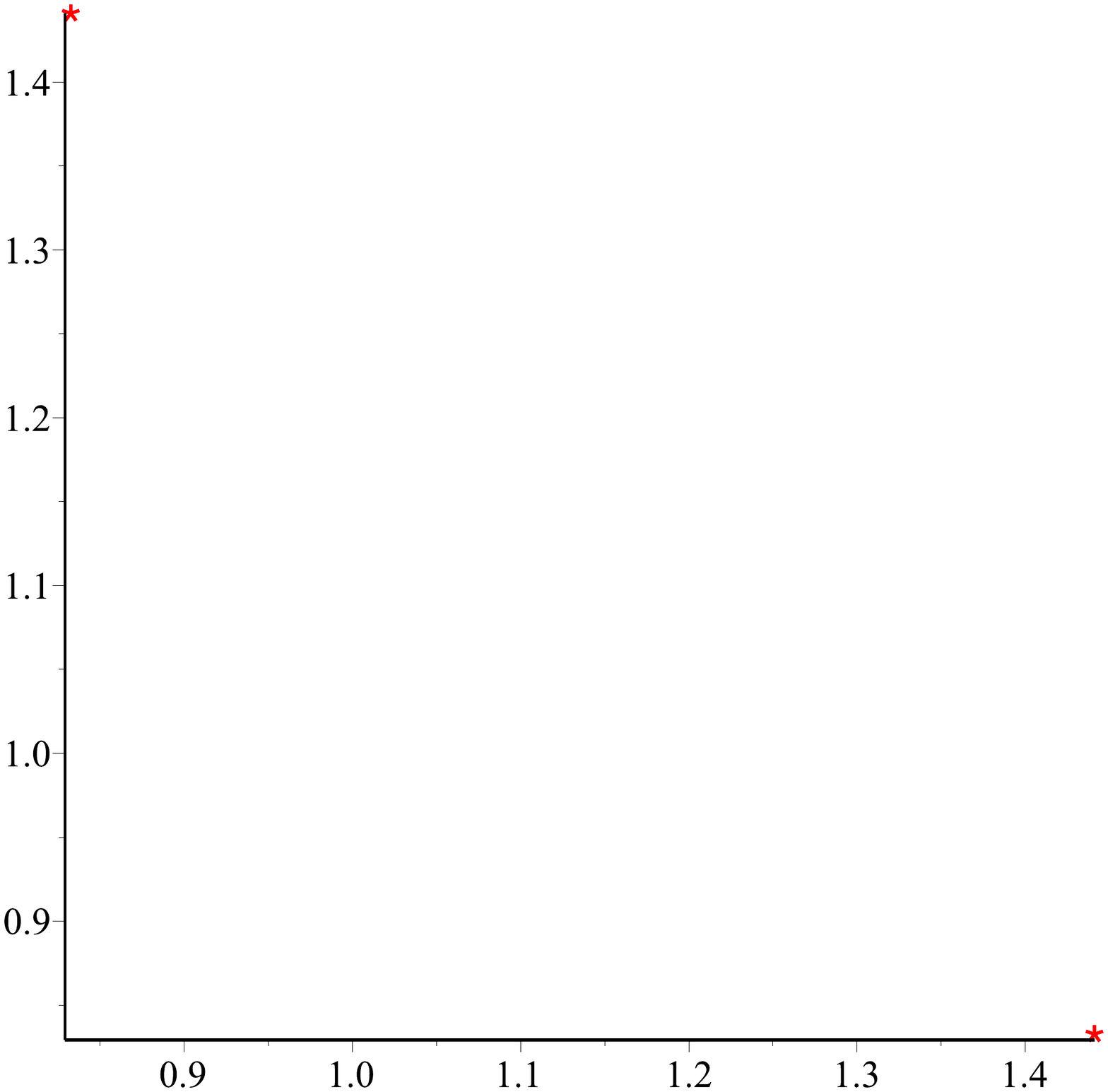}}
\caption{2-cycle in Holmes cubic map}
\label{holmes2}
\end{minipage}
\end{figure}

\section*{Acknowledgments}
M.T. is supported in part by the NSF grant DMS--1636435.

\end{document}